\pdfoutput=1
\documentclass[12pt]{article}

\usepackage[british]{babel}
\usepackage[T1]{fontenc}
\usepackage[utf8]{inputenc}
\usepackage{lmodern}
\usepackage[a4paper,margin=2.5cm]{geometry}
\usepackage{setspace}
\usepackage[mathlines]{lineno}
\usepackage{authblk}
\usepackage{amsmath,amssymb,amsthm}
\usepackage{bm}
\usepackage{booktabs}
\usepackage{capt-of}
\usepackage{graphicx}
\usepackage{pdfpages}
\usepackage{xcolor}
\usepackage{tcolorbox}
\usepackage{tikz}
\usetikzlibrary{positioning,arrows.meta,backgrounds,calc}
\usepackage[natbibapa]{apacite}
\AtBeginDocument{}
\AtBeginDocument{}

\AtBeginDocument{\renewcommand{\doi}[1]{\url{https://doi.org/#1}}}
\usepackage[hidelinks]{hyperref}

\graphicspath{{figures/}}
\newtheorem{proposition}{Proposition}

\newcommand{\BookDocs}{2,063}
\newcommand{\BookAnalysed}{2,004}
\newcommand{\BookRhoRandom}{\ensuremath{0.17}}

\newcommand{\BookPairsBoth}{\ensuremath{7.4}}
\newcommand{\BookRoots}{16}
\newcommand{\BookThemes}{28}
\newcommand{\BookLinks}{36}

\newcommand{\SimRuns}{54}
\newcommand{\SimMinN}{279}
\newcommand{\SimMaxN}{1,800}
\newcommand{\SimARIRootsMin}{\ensuremath{0.97}}
\newcommand{\SimARIThemesHigh}{\ensuremath{0.98}}
\newcommand{\SimARIThemesLowBal}{\ensuremath{0.65}}
\newcommand{\SimARIThemesLowUnb}{\ensuremath{0.54}}
\newcommand{\SimTFiveUnbHigh}{78\%}
\newcommand{\SimTSevenMin}{67\%}
\newcommand{\SimFalseBranchRoots}{4}
\newcommand{\SimFalseBranchOf}{18}
\newcommand{\SimFalseBranchUnb}{7\%}
\newcommand{\SimPrecBranchMin}{\ensuremath{0.35}}
\newcommand{\SimPrecBranchMax}{\ensuremath{0.45}}
\newcommand{\SimPrecConvMin}{\ensuremath{0.10}}
\newcommand{\SimPrecConvMax}{\ensuremath{0.32}}
\newcommand{\SimWeakNMin}{5}
\newcommand{\SimWeakNMax}{11}
\newcommand{\SimWeakResMin}{\ensuremath{2.9}}
\newcommand{\SimWeakResMax}{\ensuremath{6.9}}
\newcommand{\SimFalseARIThemesMin}{\ensuremath{0.47}}
\newcommand{\SimFalseARIThemesMax}{\ensuremath{0.56}}
\newcommand{\RuleAndFalseUnb}{4\%}
\newcommand{\RuleAndLinks}{19}
\newcommand{\RuleAndBranching}{5}
\newcommand{\RuleAndConvergent}{5}
\newcommand{\RuleAdoptedLinks}{36}
\newcommand{\RuleAdoptedBranching}{9}
\newcommand{\RuleAdoptedConvergent}{12}
\newcommand{\ValRuns}{72}
\newcommand{\ValBranching}{254}
\newcommand{\ValFalse}{20}
\newcommand{\ValFalseShareAll}{8\%}
\newcommand{\ValPersistent}{211}
\newcommand{\ValFalsePersistent}{5}
\newcommand{\ValFalseSharePersistent}{2.4\%}
\newcommand{\ValFalseShareNotPersistent}{35\%}
\newcommand{\ValFalseRemoved}{75\%}
\newcommand{\ValTrueLost}{12\%}

\newcommand{\DevPersistent}{173}

\newcommand{\BookTopicCoverage}{98\%}
\newcommand{\BookRootOneSize}{434}
\newcommand{\BookBranching}{9}
\newcommand{\BookConvergent}{12}
\newcommand{\BookBranchingPersist}{4}
\newcommand{\BookConvergentPersist}{10}
\newcommand{\BookLinksPersist}{24}

\newcommand{\BookAgreeRoots}{\ensuremath{0.74}}
\newcommand{\BookAgreeThemes}{\ensuremath{0.45}}
\newcommand{\BookPairs}{120}
\newcommand{\BookPairsBothAreas}{26}
\newcommand{\BookPairsTopOnly}{20}
\newcommand{\BookPairsRefOnly}{9}
\newcommand{\BookPairsNeither}{65}
\newcommand{\EvoPairs}{78}
\newcommand{\EvoPairsThree}{15}
\newcommand{\EvoPairsMove}{34}
\newcommand{\EvoConverging}{19}
\newcommand{\EvoConsolidating}{16}
\newcommand{\EvoDiverging}{4}
\newcommand{\EvoDriftingApart}{24}
\newcommand{\EvoCloserinReferences}{2}
\newcommand{\EvoFartherinReferences}{9}
\newcommand{\EvoStable}{4}
\newcommand{\XaiDocs}{3,628}

\newcommand{\XaiRho}{\ensuremath{0.04}}

\newcommand{\DisDocs}{2,272}
\newcommand{\DisAnalysed}{1,521}

\newcommand{\StabConsRootsMin}{\ensuremath{0.82}}
\newcommand{\StabConsRootsMax}{\ensuremath{0.94}}
\newcommand{\StabConsThemesMin}{\ensuremath{0.76}}
\newcommand{\StabConsThemesMax}{\ensuremath{0.80}}
\newcommand{\StabSingleRootsMin}{\ensuremath{0.47}}
\newcommand{\StabSingleRootsMax}{\ensuremath{0.76}}
\newcommand{\StabSingleThemesMin}{\ensuremath{0.43}}
\newcommand{\StabSingleThemesMax}{\ensuremath{0.44}}
\newcommand{\SensKFiveRoots}{\ensuremath{0.54}}
\newcommand{\SensKTwentyRoots}{\ensuremath{0.74}}
\newcommand{\SensKFiveThemes}{\ensuremath{0.28}}
\newcommand{\SensKTwentyThemes}{\ensuremath{0.39}}
\newcommand{\SensAssocThemes}{9}
\newcommand{\SensAssocNMI}{\ensuremath{0.14}}
\newcommand{\SensTopRoots}{\ensuremath{0.44}}
\newcommand{\SensIdfRoots}{\ensuremath{0.35}}
\newcommand{\SensTiabThemes}{\ensuremath{0.12}}
\newcommand{\SensNoCleanThemes}{\ensuremath{0.28}}

\newcommand{\BookProfiledRoots}{7}

\title{Multiplex coupling: Connecting the intellectual base and the topics of a research field}

\author[1,2]{Massimo Aria}
\author[1,2,3]{Corrado Cuccurullo}
\author[1,2]{Maria Spano}
\affil[1]{Department of Economics and Statistics, University of Naples Federico II, Naples, Italy}
\affil[2]{K-Synth academic Spinoff, University of Naples Federico II, Naples, Italy}
\affil[3]{Department of Economics and Management, University of Campania Luigi Vanvitelli, Capua, Italy}
\date{}

\begin{document}

\maketitle
\begin{center}
\small Corresponding author: Massimo Aria, \href{mailto:aria@unina.it}{aria@unina.it}\\[2pt]
ORCID: Massimo Aria 0000-0002-8517-9411; Corrado Cuccurullo 0000-0002-7401-8575;\\
Maria Spano 0000-0002-3103-2342
\end{center}

\begin{abstract}
\noindent
Two documents can be close because they cite the same works or because they write
about the same things, and the two kinds of proximity often disagree. Science
mapping reads them in separate maps or merges them, as hybrid clustering does;
either way the disagreement is lost. We introduce multiplex coupling, a
framework that couples the same documents by their references and by their
terms, keeps the two layers separate and analyses their relation. Similarities are standardised against a null model, and the agreement
between the layers is measured on random pairs of documents. Consensus clustering
partitions each layer into roots, groups sharing an intellectual base, and
themes, groups sharing a topic. A test on the roots $\times$ themes table
identifies branching roots and convergent themes, and pairs of roots are followed
over time. Every structure is reported with its persistence across numbers of
neighbours, which on synthetic data identifies most false branching.
Applications to Web of Science, Scopus and OpenAlex collections show that the
two layers carry largely complementary information, and reveal substantive and
methodological bases that have grown into several topics. The framework is
implemented in \textit{bibliometrix}.
\end{abstract}


\clearpage
\doublespacing

\section{Introduction}\label{sec:intro}

Two scientific documents can be close in two ways: they can cite the same works,
or they can write about the same things. Science mapping has a method for each.
Bibliographic coupling \citep{kessler1963} relates documents through their shared
references and groups them into research fronts defined by a common intellectual
base \citep{persson1994,boyack2010}; co-word analysis \citep{callon1983} and its
variants relate them through their shared terms and describe the conceptual
structure of a field. In the workflow of \citet{aria2017}, the two methods map
different knowledge structures of the same collection, and the analyst reads one
map at a time.

The two kinds of proximity need not coincide, and often they do not. Documents
that draw on the same base may apply it to unrelated subjects, as when a
methodological literature is cited across a whole discipline; documents that
cite different bases may meet on the same subject. Comparisons of citation-based
and text-based maps of the same documents find that the two partitions overlap
little \citep{boyack2020,xie2025}. A map built on one kind of proximity cannot
show where it disagrees with the other.

The established answer to this disagreement is hybrid clustering, which combines
citation links and lexical similarity into a single measure and clusters the
documents once \citep{janssens2008,liu2010,thijs2013,thijs2018}. Hybrid
measures have been found to give more accurate partitions than either signal
alone \citep{liu2010,boyack2020}, which is what matters for delineating topics.
For understanding how a field is structured, however, merging the two signals
discards the information carried by their disagreement. A hybrid partition can
no longer tell whether a group of documents shares references but not topics,
or topics but not references.

Multiplex coupling combines the two kinds of layer without merging them.
Comparing two layers on the same documents raises methodological problems that
single-layer mapping does not face. Similarities in the two layers are on
different scales and have to be standardised before they can be compared.
Clusters found by modularity optimisation
change with the random order of the algorithm and with the density of the graph,
so a relation found between the clusters of the two layers may be an artefact of
an unstable partition.
We address, in particular, three research questions:

\begin{description}
  \item[RQ1 --] How far do the intellectual base and the topics of a field agree on
  the same documents, and how can their agreement be measured without the bias
  introduced by the selection of neighbours?
  \item[RQ2 --] Which intellectual roots have grown into several themes, which
  themes have been reached from different roots, and how can these structures be
  distinguished from artefacts of the clustering?
  \item[RQ3 --] How do the relations between roots, in references and in topics,
  change over time?
\end{description}

We answer these questions with a framework that treats the two kinds of
proximity as two layers of a multiplex network on the same documents
\citep{kivela2014,battiston2014}. Each layer is standardised against a null
model, so that similarities in references and in terms become comparable, and
the agreement between the layers is measured on random pairs of documents and
tested by permutation. Consensus clustering then partitions the references layer
into \emph{roots}, groups of documents that share an intellectual base, and the
topic layer into \emph{themes}, groups that share a topic. A single rule on
their cross-tabulation tells which roots have grown into several themes and
which themes have been reached from several roots, and pairs of roots are
placed on the plane of their proximity in references and in topics and followed
over time. Robustness is part of the result rather than a check made
afterwards: every link and structure is reported with its persistence across
numbers of neighbours, which on synthetic data turns out to be the diagnostic of
false branching. We apply the framework to a collection of management studies
from Web of Science, report two further collections from Scopus and OpenAlex in
the Supplementary Material, assess it on synthetic collections with a planted
structure, and make it available in \textit{bibliometrix} and Biblioshiny.

Section~\ref{sec:background} reviews the background and
Section~\ref{sec:method} presents the framework. Section~\ref{sec:application}
applies it to the management collection and Section~\ref{sec:validation}
assesses it on synthetic data. Section~\ref{sec:discussion} discusses results
and limitations, and Section~\ref{sec:conclusions} concludes.

\section{Background}\label{sec:background}

\subsection{Citation-based and word-based science mapping}\label{sec:bg-mapping}\label{sec:bg-hybrid}

Science mapping represents the relational structure of a research field from
bibliographic data \citep{borner2003,cobo2011,zupic2015}. Two of its oldest
methods use the references of documents. In bibliographic coupling
\citep{kessler1963} two documents are related when they cite the same works; in
co-citation analysis \citep{small1973} two works are related when they are cited
together. \citet{persson1994} gave the two methods distinct objects: the
co-cited works form the \emph{intellectual base} of a field, and the citing
documents grouped by coupling form its \emph{research fronts}. Among the
citation approaches, bibliographic coupling has been found to represent the
research front of a set of documents more accurately than co-citation
\citep{boyack2010}, while direct citation gives the most accurate taxonomies of
large literatures \citep{klavans2017}.

Co-word analysis \citep{callon1983,callon1991} relates terms that occur in the
same documents, and the clusters of the resulting network are read as the themes
of a field, its conceptual structure. With a Document $\times$ Attribute matrix
$\mathbf{A}$, coupling is $\mathbf{A}\mathbf{A}^{\top}$, whereas co-citation and
co-word analysis are $\mathbf{A}^{\top}\mathbf{A}$ \citep{aria2017}. Coupling can be applied to attributes other than
references: author keyword coupling relates authors through the keywords they
use \citep{yang2016}, and in the same way documents can be coupled through their
terms. The two layers of the present framework are of this kind, two couplings of
the same documents, one by references and one by terms. Science mapping
software computes these structures separately \citep{cobo2011,aria2017}, and
their comparison is left to the analyst.

The relation between citation-based and word-based structures has been studied
since \citet{braam1991}, who characterised co-citation clusters by the words of the
documents citing them. They found that co-citation analysis does identify
subject-matter specialties, but that different clusters may identify the same
specialty. Later work moved from comparing the two structures to combining them. Hybrid
clustering combines citation links and lexical similarity into a single measure,
either by a weighted combination of the two similarity matrices or by integrating
the partitions obtained from each \citep{janssens2008,liu2010}. The approach has
been refined with second-order similarities \citep{thijs2013}, used to detect
emerging topics \citep{glanzel2012}, and assessed for the contribution of its
lexical component \citep{thijs2018}. Citation-based, text-based and combined
relatedness measures have been compared experimentally \citep{ahlgren2009}, and
then on large collections against principled accuracy criteria
\citep{waltman2020,ahlgren2020}. On nine million biomedical documents,
\citet{boyack2020} found that a hybrid of direct citation and textual relatedness
gives the most accurate clusters, and that the clusters of the citation-based and
text-based models differ in content more than in accuracy: the two logics are
complementary.

A parallel line of work compares the two kinds of structure without merging them.
At the level of institutions, \citet{yan2012} found that coupling, citation and
co-citation networks resemble each other, as do topical and co-word networks, and
that the main dimension separating the networks is citation-based versus
non-citation-based; they recommended hybrid or heterogeneous networks for the
study of scholarly communication. \citet{bascur2025} found that the same topic
categories are well or poorly represented in maps built from citation and from
text similarity, with citation networks performing better for some categories and
for small clusters. \citet{xie2025} compared citation-based clustering with topic
modelling on cardiovascular research and found the relations between clusters and
topics generally weak: citation clusters reflect the intellectual structure and
its micro-communities, while topics better capture the societal needs related to the disease.

These comparisons establish that citation-based and word-based partitions of the
same documents differ, and differ systematically. They measure the difference
for the map as a whole or for categories of topics. They do not identify which
clusters of one partition are related to which clusters of the other beyond
chance, nor which citation clusters are spread over several topics.

\subsection{Multiplex and multilayer networks}\label{sec:bg-multiplex}

A multiplex network is a set of nodes connected by several types of edges, one
layer for each type; it is the special case of a multilayer network in which every
layer contains the same nodes \citep{kivela2014}. Several measures describe how a
node or an edge is distributed across layers. The participation coefficient, used
in Section~\ref{sec:nodes}, measures how evenly the edges of a node are spread over
the layers, and edge overlap measures how many pairs of nodes are connected in more
than one layer \citep{battiston2014}. At the level of the whole network,
\citet{dedomenico2015} asked how many layers are needed to represent a system and
proposed to aggregate layers that carry redundant information.

Community detection in multiplex networks usually looks for communities across
layers. In the multislice modularity of \citet{mucha2010}, each node is connected to
its copies in the other layers, and a coupling parameter sets how strongly the
partition of one layer is tied to that of the others. This is the multilayer
counterpart of hybrid clustering: the layers inform a partition, and their
differences are a nuisance to be controlled.
\citet{aria2024comp} compared the conceptual structures of different units of
observation, such as institutions, in a joint representation; the comparison is
between groups of documents on the same layer, not between layers on the same
documents.

Existing approaches either merge the citation and the lexical signal into one
partition, as hybrid clustering and multislice community detection do, or compare
the two signals globally, as the comparisons of networks and maps do. What is
missing is a method that partitions the two layers of the same documents
separately, tests which clusters of one are related to which clusters of the
other, classifies the resulting structures, and reports how robust they are and how
they change over time. The framework of Section~\ref{sec:method} addresses this gap.

\section{The multiplex coupling framework}\label{sec:method}

Let $\mathcal{D}=\{d_1,\dots,d_N\}$ denote a collection of documents, each
described by the references it cites and by the terms it uses. The framework
represents $\mathcal{D}$ as a multiplex network with two layers on the same
nodes. In the references layer two documents are close when they cite the same
works; in the topic layer, when they use the same terms. Both layers are products
$\mathbf{A}\mathbf{A}^{\top}$ of a Document $\times$ Attribute matrix
\citep{aria2017,aria2026book}. The references layer is bibliographic coupling
itself \citep{kessler1963}; the topic layer couples documents through their
terms.

The analytical goal is threefold. First, to measure how far the two layers agree
on the same documents. Second, to partition each layer into clusters, the roots
of the references layer and the themes of the topic layer, and to identify which
roots feed which themes. Third, to compare the roots with each other and to
follow how their relations change over time. The three goals correspond to the
three phases of Fig.~\ref{fig:pipeline}; Table~\ref{tab:symbols} lists the
mathematical notation.

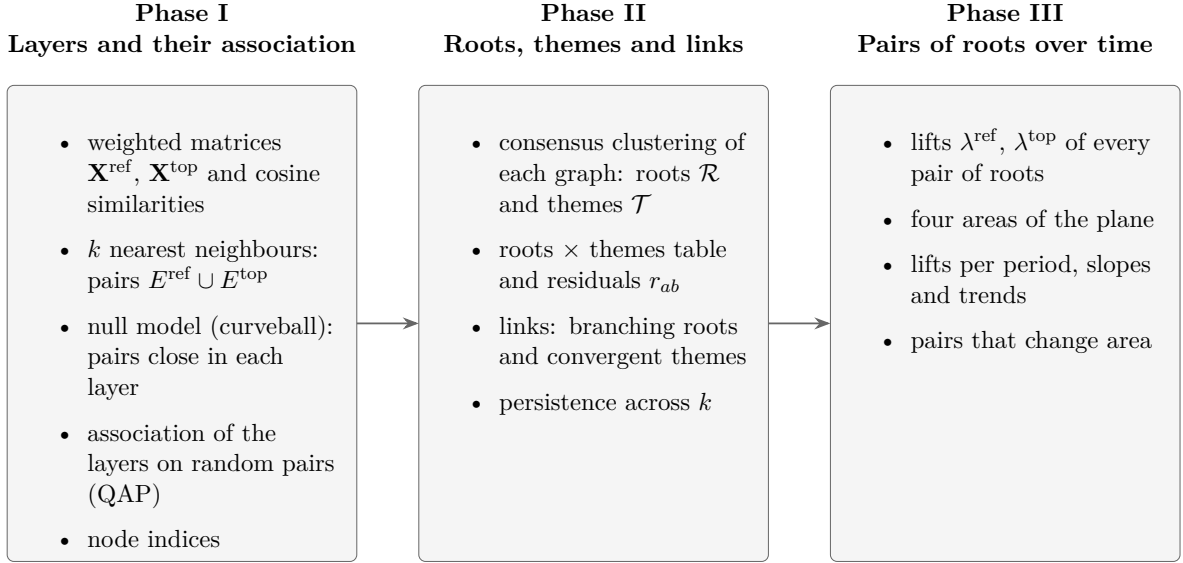
\begin{figure}[t]
  \centering
  {\linespread{1}\selectfont\resizebox{\textwidth}{!}{\begin{tikzpicture}[
  font=\small,
  content/.style={inner sep=6pt, text width=5.0cm, align=left, anchor=north},
  box/.style={draw=black!60, rounded corners=3pt, fill=black!4},
  head/.style={font=\small\bfseries, align=center, text width=5.6cm, anchor=north},
  arrow/.style={-{Stealth[length=2.5mm]}, thick, black!60}
]
\newcommand{\phaselist}{\begin{itemize}\setlength{\leftmargin}{1em}\setlength{\itemsep}{2pt}%
  \setlength{\parsep}{0pt}\setlength{\topsep}{0pt}\setlength{\partopsep}{0pt}\raggedright}
\node[head] (h1) at (0,0) {Phase I\\Layers and their association};
\node[head] (h2) at (6.4,0) {Phase II\\Roots, themes and links};
\node[head] (h3) at (12.8,0) {Phase III\\Pairs of roots over time};
\path let \p1 = (h1.south), \p2 = (h2.south), \p3 = (h3.south),
  \n1 = {min(\y1, \y2, \y3)} in coordinate (top) at (0, \n1 - 3mm);
\node[content] (p1) at (h1 |- top) {\phaselist
  \item weighted matrices $\mathbf{X}^{\mathrm{ref}}$, $\mathbf{X}^{\mathrm{top}}$ and cosine similarities
  \item $k$ nearest neighbours: pairs $E^{\mathrm{ref}}\cup E^{\mathrm{top}}$
  \item null model (curveball): pairs close in each layer
  \item association of the layers on random pairs (QAP)
  \item node indices
\end{itemize}};
\node[content] (p2) at (h2 |- top) {\phaselist
  \item consensus clustering of each graph: roots $\mathcal{R}$ and themes $\mathcal{T}$
  \item roots $\times$ themes table and residuals $r_{ab}$
  \item links: branching roots and convergent themes
  \item persistence across $k$
\end{itemize}};
\node[content] (p3) at (h3 |- top) {\phaselist
  \item lifts $\lambda^{\mathrm{ref}}$, $\lambda^{\mathrm{top}}$ of every pair of roots
  \item four areas of the plane
  \item lifts per period, slopes and trends
  \item pairs that change area
\end{itemize}};
\path let \p1 = (p1.south), \p2 = (p2.south), \p3 = (p3.south),
  \n1 = {min(\y1, \y2, \y3)} in coordinate (bottom) at (0, \n1);
\begin{scope}[on background layer]
  \foreach \p in {p1, p2, p3}
    \draw[box] (\p.north west) rectangle (\p.east |- bottom);
\end{scope}
\coordinate (mid) at ($(top)!0.5!(bottom)$);
\draw[arrow] (p1.east |- mid) -- (p2.west |- mid);
\draw[arrow] (p2.east |- mid) -- (p3.west |- mid);
\end{tikzpicture}}}
  \caption{The multiplex coupling framework. Phase I builds the references
  layer and the topic layer on the same documents, keeps the nearest neighbours
  of every document in each layer and standardises each pair against a null
  model. Phase II clusters each layer into roots and themes and links them
  through the roots $\times$ themes table. Phase III places every pair of roots
  on the plane of their proximity in references and in topics, and follows it
  over time.}
  \label{fig:pipeline}
\end{figure}

\begin{table}[t]
  \centering
  \caption{Mathematical notation.}
  \label{tab:symbols}
  \small
  \begin{tabular}{@{}ll@{}}
    \toprule
    Notation & Meaning \\
    \midrule
    $\mathcal{D}=\{d_1,\dots,d_N\}$ & documents analysed \\
    $\mathbf{A}^{\mathrm{ref}}$ ($N\times M$) & Document $\times$ Cited reference, binary \\
    $\mathbf{A}^{\mathrm{top}}$ ($N\times L$) & Document $\times$ Term \\
    $\mathbf{X}^{\mathrm{ref}},\ \mathbf{X}^{\mathrm{top}}$ & the same matrices, weighted \\
    $s^{\mathrm{ref}}_{ij},\ s^{\mathrm{top}}_{ij}$ & similarity of $d_i$ and $d_j$ in the two layers \\
    $E^{\mathrm{ref}},\ E^{\mathrm{top}}$ & pairs of nearest neighbours in the two layers \\
    $k$ & number of nearest neighbours \\
    $\mathcal{E}=E^{\mathrm{ref}}\cup E^{\mathrm{top}}$ & pairs analysed \\
    $B$, $\alpha$ & randomised matrices and level of the null model \\
    $\mathcal{R}=\{R_1,\dots,R_p\}$ & roots, clusters of the references layer \\
    $\mathcal{T}=\{T_1,\dots,T_q\}$ & themes, clusters of the topic layer \\
    $\mathbf{N}=[n_{ab}]$ & documents of root $R_a$ in theme $T_b$ \\
    $r_{ab}$ & adjusted standardised residual of $n_{ab}$ \\
    $z$, $n_{\min}$ & thresholds of the link rule on $r_{ab}$ and $n_{ab}$ \\
    $\kappa^{\mathrm{ref}}(C),\ \kappa^{\mathrm{top}}(C)$ & cohesion of a cluster $C$ \\
    $\lambda^{\mathrm{ref}}(R_a,R_b),\ \lambda^{\mathrm{top}}(R_a,R_b)$ & lift between two roots \\
    $\{t_1,\dots,t_P\}$, $\mathcal{D}^{(t)}$ & periods and their documents \\
    $m$ & threshold on the slopes of the trends \\
    \bottomrule
  \end{tabular}
\end{table}

\subsection{Documents and the two layers}\label{sec:layers}

Let $\mathbf{A}^{\mathrm{ref}}$ be a Document $\times$ Reference binary
matrix with $a^{\mathrm{ref}}_{im}=1$ when $d_i$ cites reference $m$. Only the
references cited by at least two documents are kept, since a reference cited
once links no pair.

The topic layer is built from keywords or from the terms of titles and
abstracts. Let $\mathbf{A}^{\mathrm{top}}$ be the corresponding Document $\times$
Term matrix, with $a^{\mathrm{top}}_{il}=\mathrm{tf}_{il}$, the number of
occurrences of term $l$ in $d_i$; for keywords, $a^{\mathrm{top}}_{il}$ is binary.
Keywords enter as binary indicators, terms of titles and abstracts with the weight
$1+\ln \mathrm{tf}_{il}$ when $\mathrm{tf}_{il}>0$. Every term is then weighted by its inverse
document frequency $\ln(N/\mathrm{df}_l)$: sharing a term that most documents
use counts for little, sharing a rare one for much more. References are not
weighted.

The similarity of two documents in a layer is the cosine of their weighted
rows \citep{salton1983},
\begin{equation}\label{eq:cosine}
  s_{ij} = \frac{\mathbf{x}_i^{\top}\mathbf{x}_j}{\lVert\mathbf{x}_i\rVert\,\lVert\mathbf{x}_j\rVert},
\end{equation}
computed separately in the two layers. Association strength and the Jaccard
index could be used as alternatives \citep{vaneck2009}.

\paragraph{Nearest neighbours}
A collection of $N$ documents has $N(N-1)/2$ pairs, and most of them are
unrelated in both layers. Each layer is therefore reduced to the $k$ nearest
neighbours of every document, the $k$ documents with the
highest positive similarity to $d_i$. The relation is made symmetric by union:
a pair belongs to $E^{\mathrm{ref}}$ when either document is among the $k$
nearest neighbours of the other in the references layer, and likewise for
$E^{\mathrm{top}}$. The pairs analysed are
$\mathcal{E}=E^{\mathrm{ref}}\cup E^{\mathrm{top}}$, those close in at least
one layer; for each of them both similarities are computed exactly.

\subsection{Standardisation, association and node indices}\label{sec:null}

A cosine of 0.2 means one thing in a layer where documents cite forty
references and another in a layer where they carry five keywords. Before the two
layers can be compared, each similarity has to be set against what the same pair
would obtain by chance. The binary matrix of each layer is randomised with the
curveball algorithm \citep{strona2014}, which keeps the number of features of
every document and the number of documents of every feature. A document with a
long reference list keeps a long list in the randomised matrix, and a classic
cited by half the collection stays as popular; document frequencies, and with
them the idf weights, do not change.

For each pair $(i,j)\in\mathcal{E}$ and each layer, the similarity is computed
on $B$ randomised matrices, and the empirical $p$-value is
\begin{equation}\label{eq:pvalue}
  p_{ij} = \frac{1+\#\{b : s^{(b)}_{ij}\ge s_{ij}\}}{B+1}.
\end{equation}
The pair is \emph{close} in that layer when $p_{ij}<\alpha$. Crossing the two outcomes gives four situations
(Fig.~\ref{fig:situations}): \emph{consolidation} (close in both layers),
\emph{branching} (close in references only), \emph{convergence} (close in
topics only) and \emph{detachment} (close in neither). Statistical
inference is performed at cluster level.

\begin{figure}[t]
  \centering
  {\linespread{1}\selectfont\begin{tikzpicture}[
  font=\small,
  cell/.style={draw=black!50, minimum width=4.2cm, minimum height=1.9cm, align=center, text width=3.9cm}
]
\node[cell, fill=black!10] (cons) at (0,0) {\textbf{consolidation}\\same references,\\same topics};
\node[cell, fill=black!4] (conv) at (-4.2,0) {\textbf{convergence}\\different references,\\same topics};
\node[cell, fill=black!4] (bran) at (0,-1.9) {\textbf{branching}\\same references,\\different topics};
\node[cell] (det) at (-4.2,-1.9) {\textbf{detachment}\\different references,\\different topics};
\node[font=\small\itshape, above=1mm of conv] {not close in references};
\node[font=\small\itshape, above=1mm of cons] {close in references};
\node[font=\small\itshape, anchor=east, align=right] at ([xshift=-2mm]conv.west) {close\\in topics};
\node[font=\small\itshape, anchor=east, align=right] at ([xshift=-2mm]det.west) {not close\\in topics};
\end{tikzpicture}}
  \caption{The four situations of a pair of documents, by whether the pair is
  closer than chance in the references layer (columns) and in the topic layer
  (rows).}
  \label{fig:situations}
\end{figure}

\paragraph{Association between the layers}\label{sec:association}
The association between the two layers is measured on a random sample of pairs
of documents, as
Spearman's correlation between
$s^{\mathrm{ref}}_{ij}$ and $s^{\mathrm{top}}_{ij}$, and tested by a quadratic
assignment procedure \citep{hubert1976,krackhardt1988}: the documents of the
topic layer are permuted, which keeps the structure of each layer and breaks
only the correspondence between them.

\paragraph{Node indices}\label{sec:nodes}

Four indices describe each document (Fig.~\ref{fig:nodes}). The \emph{convergence index} $CI_i$ is the
share of the topic neighbours of $d_i$ that are not reference neighbours, and
the \emph{branching index} $BI_i$ is the share of its reference neighbours that
are not topic neighbours. The \emph{participation coefficient}
\citep{battiston2014}
\begin{equation}\label{eq:participation}
  P_i = 2\left[1-\left(\frac{k^{\mathrm{ref}}_i}{o_i}\right)^2-\left(\frac{k^{\mathrm{top}}_i}{o_i}\right)^2\right],
  \qquad o_i = k^{\mathrm{ref}}_i + k^{\mathrm{top}}_i,
\end{equation}
where $k^{\mathrm{ref}}_i$ and $k^{\mathrm{top}}_i$ are the degrees of $d_i$ in
the two neighbour graphs, is 1 when the document is equally connected in both
layers and 0 when it is connected in one only. The \emph{layer alignment}
$LA_i$ is Spearman's correlation between the similarities of $d_i$ to all other
documents in the two layers.

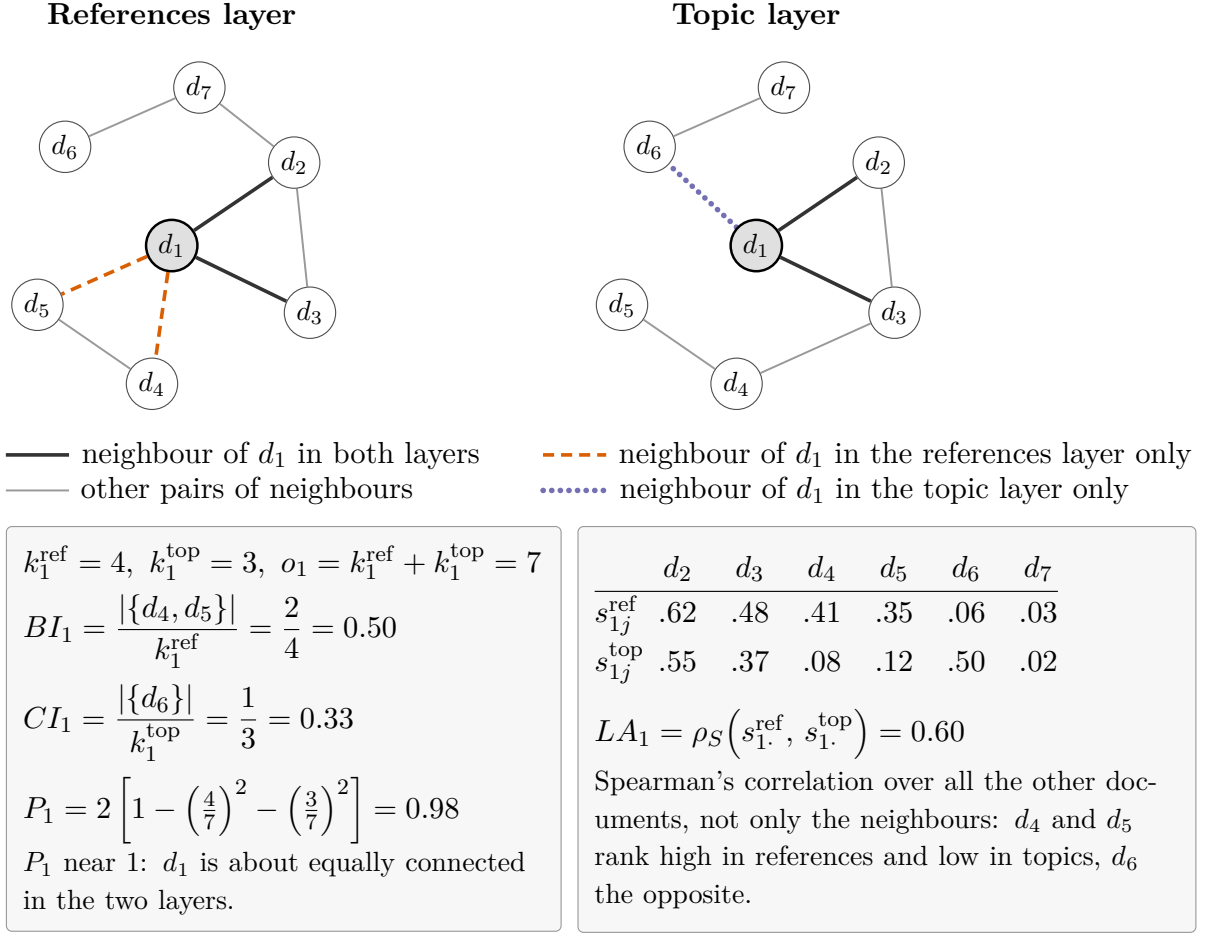
\begin{figure}[!tbp]
  \centering
  {\linespread{1}\selectfont\resizebox{\textwidth}{!}{\begin{tikzpicture}[
  font=\small,
  doc/.style={circle, draw=black!70, fill=white, inner sep=1.2pt, minimum size=6.5mm, font=\footnotesize},
  focal/.style={doc, fill=black!12, draw=black, line width=0.9pt},
  both/.style={line width=1.3pt, black!80},
  refonly/.style={line width=1.3pt, draw={rgb,255:red,217;green,95;blue,2}, dash pattern=on 4pt off 2.5pt},
  toponly/.style={draw={rgb,255:red,117;green,112;blue,179}, dash pattern=on 0.1pt off 3pt, line cap=round, line width=2pt},
  other/.style={line width=0.7pt, black!40},
  panel/.style={font=\small\bfseries}
]
\newcommand{\docs}[1]{%
  \node[focal] (#1d1) at (0,0) {$d_1$};
  \node[doc] (#1d2) at (1.55,1.05) {$d_2$};
  \node[doc] (#1d3) at (1.75,-0.85) {$d_3$};
  \node[doc] (#1d4) at (-0.25,-1.75) {$d_4$};
  \node[doc] (#1d5) at (-1.7,-0.75) {$d_5$};
  \node[doc] (#1d6) at (-1.35,1.25) {$d_6$};
  \node[doc] (#1d7) at (0.35,2.0) {$d_7$};
}

\begin{scope}
  \docs{r}
  \node[panel] at (0,2.9) {References layer};
  \begin{scope}[on background layer]
    \draw[other] (rd2) -- (rd3); \draw[other] (rd4) -- (rd5); \draw[other] (rd6) -- (rd7);
    \draw[other] (rd2) -- (rd7);
    \draw[both] (rd1) -- (rd2); \draw[both] (rd1) -- (rd3);
    \draw[refonly] (rd1) -- (rd4); \draw[refonly] (rd1) -- (rd5);
  \end{scope}
\end{scope}

\begin{scope}[xshift=7.4cm]
  \docs{t}
  \node[panel] at (0,2.9) {Topic layer};
  \begin{scope}[on background layer]
    \draw[other] (td2) -- (td3); \draw[other] (td6) -- (td7); \draw[other] (td4) -- (td5);
    \draw[other] (td3) -- (td4);
    \draw[both] (td1) -- (td2); \draw[both] (td1) -- (td3);
    \draw[toponly] (td1) -- (td6);
  \end{scope}
\end{scope}

\begin{scope}[yshift=-2.65cm, xshift=-1.9cm]
  \draw[both] (-0.2,0) -- (0.6,0) node[right, black] {neighbour of $d_1$ in both layers};
  \draw[other] (-0.2,-0.45) -- (0.6,-0.45) node[right, black] {other pairs of neighbours};
  \draw[refonly] (6.6,0) -- (7.4,0) node[right, black] {neighbour of $d_1$ in the references layer only};
  \draw[toponly] (6.6,-0.45) -- (7.4,-0.45) node[right, black] {neighbour of $d_1$ in the topic layer only};
\end{scope}

\node[draw=black!40, rounded corners=2pt, fill=black!3, align=left, anchor=north west,
      text width=6.6cm, inner sep=6pt] (left box) at (-2.1,-3.55) {%
  $k^{\mathrm{ref}}_1=4,\ k^{\mathrm{top}}_1=3,\ o_1=k^{\mathrm{ref}}_1+k^{\mathrm{top}}_1=7$\\[4pt]
  $BI_1 = \dfrac{|\{d_4,d_5\}|}{k^{\mathrm{ref}}_1} = \dfrac{2}{4} = 0.50$\\[5pt]
  $CI_1 = \dfrac{|\{d_6\}|}{k^{\mathrm{top}}_1} = \dfrac{1}{3} = 0.33$\\[5pt]
  $P_1 = 2\left[1-\left(\tfrac{4}{7}\right)^2-\left(\tfrac{3}{7}\right)^2\right] = 0.98$\\[2pt]
  {\footnotesize $P_1$ near 1: $d_1$ is about equally connected in the two layers.}};

\path let \p1 = ($(left box.north) - (left box.south)$) in
  node[draw=black!40, rounded corners=2pt, fill=black!3, align=left, anchor=north west,
       text width=7.4cm, inner sep=6pt, minimum height=\y1]
       at ([xshift=2mm]left box.north east) {%
  \begin{tabular}{@{}l@{\hspace{5pt}}cccccc@{}}
    & $d_2$ & $d_3$ & $d_4$ & $d_5$ & $d_6$ & $d_7$\\[1pt]
    \hline
    \rule{0pt}{11pt}$s^{\mathrm{ref}}_{1j}$ & .62 & .48 & .41 & .35 & .06 & .03\\[3pt]
    $s^{\mathrm{top}}_{1j}$ & .55 & .37 & .08 & .12 & .50 & .02\\
  \end{tabular}\\[8pt]
  $LA_1 = \rho_S\!\left(s^{\mathrm{ref}}_{1\cdot},\, s^{\mathrm{top}}_{1\cdot}\right) = 0.60$\\[4pt]
  {\footnotesize Spearman's correlation over all the other documents, not only the
  neighbours: $d_4$ and $d_5$ rank high in references and low in topics, $d_6$ the
  opposite.}};
\end{tikzpicture}}}
  \caption{The four node indices on a small example. The seven documents are in
  the same position in the two layers; the edges are the pairs of nearest
  neighbours. Those of the focal document $d_1$ are solid when $d_1$ and its
neighbour are close in both layers, dashed when they are close in references
only and dotted when they are close in topics only. Two of
  its four reference neighbours are not topic neighbours ($BI_1=0.50$), one of its
  three topic neighbours is not a reference neighbour ($CI_1=0.33$), and its
  neighbours are split almost evenly between the layers ($P_1=0.98$). The layer
  alignment $LA_1$ uses the similarities of $d_1$ to all the other documents, not
  only to its neighbours.}
  \label{fig:nodes}
\end{figure}

\subsection{Roots and themes}\label{sec:clusters}

Each neighbour graph, weighted by the similarities, is partitioned by the
Louvain algorithm \citep{blondel2008}; other community detection algorithms,
such as Leiden \citep{traag2019}, could be used. These general-purpose algorithms were not designed to
reconstruct topics, and their results need to be checked rather than taken as
given \citep{held2022}. Their clusters also change with the random order in
which nodes are visited, so each partition is the consensus of
several runs. Following \citet{lancichinetti2012}, at each
iteration every edge is weighted by the share of runs that put its two documents
in the same cluster, the edges below one half are dropped, and the graph is
clustered again until all runs agree. The consensus is computed on the edges of
the neighbour graph only. The mean adjusted Rand index \citep{hubert1985} between single
runs measures how stable the partition of each layer is.

\paragraph{Why ``roots''}
The clusters of the references layer are called \emph{roots}, for intellectual
roots. In the classical reading of citation analysis the works cited by a group
of documents are the roots on which it builds: the shared exemplars through
which a community recognises its problems and methods \citep{kuhn1970}, the
intellectual base that \citet{persson1994} distinguishes from the research front
citing it. \citet{havemann2021} delimits topics by clustering the citation links
to highly cited sources; a root is instead a group of citing documents, a research
front \citep{boyack2010} named after the works it cites in common. It is not a
research school, a group held together by training, institutions and a leader
\citep{morrell1972,geison1981}: coupling data observe none of these ties.

\paragraph{Why ``themes''}
The clusters of the topic layer are called \emph{themes}, for conceptual
themes. In co-word analysis the clusters of the network of terms are read as the
themes of a field, which together form its conceptual structure
\citep{callon1991,aria2017}. A theme of the framework is a group of documents
rather than of terms: the documents that use the same concepts, whatever they
cite. Roots and themes are thus two readings of the same documents, by the base
they draw on and by the concepts they use.

Each root is named by the titles of its strong references. The strength of
reference $m$ in root $R_a$ is $n_{am}^2/n_m$, where $n_{am}$ is the number of
documents of $R_a$ that cite it and $n_m$ the number of documents that cite it
in the whole collection. The strength is the product
$n_{am}\cdot(n_{am}/n_m)$ of how many documents of the root cite the reference
and of the share of its citations that come from the root: a strong reference is
cited by many documents of the root and by few outside it. The terms of the
titles of the strongest references, weighted by their strength and by an inverse frequency across
roots, give the name. Titles come from the collection itself, from the
reference strings, or from an external source such as OpenAlex
\citep{priem2022}. Themes are named by their most characteristic terms, with
the same score.

\subsection{Links between roots and themes}\label{sec:links}

Roots and themes are two partitions of the same documents. Their
cross-tabulation $\mathbf{N}=[n_{ab}]$ counts the documents of root $R_a$ in
theme $T_b$, over the documents that have both. A root and a theme are
\emph{linked} when the root contributes to the theme more documents than
expected from the sizes of the two, judged by the adjusted standardised residual
\citep{haberman1973}
\begin{equation}\label{eq:residual}
  r_{ab} = \frac{n_{ab}-e_{ab}}{\sqrt{e_{ab}\,(1-n_{a\cdot}/n)(1-n_{\cdot b}/n)}},
  \qquad e_{ab}=\frac{n_{a\cdot}\,n_{\cdot b}}{n},
\end{equation}
where $n_{a\cdot}$ and $n_{\cdot b}$ are the row and column totals and $n$ the
grand total. The link rule is $r_{ab}>z$ and $n_{ab}\ge n_{\min}$. The
residual, approximately standard normal under independence, ensures that the
overlap is not an effect of the sizes of the root and the theme; the minimum
count keeps out overlaps of very few documents, which can reach a high residual
when the root is small. A single rule, read by rows and by columns, classifies
both sides:
\begin{itemize}
  \item a \emph{branching root} is linked to two or more themes, a \emph{convergent theme} to two or more roots;
  \item a root or a theme with one link is \emph{consolidated};
  \item one without links is \emph{dispersed}.
\end{itemize}
The agreement of the two partitions as a whole is summarised by the normalised
mutual information \citep{danon2005} and the adjusted Rand index
\citep{hubert1985}.

\begin{tcolorbox}[colback=black!3, colframe=black!45, boxrule=0.5pt, arc=2pt,
  title={Box 1. Illustrative example: from the roots $\times$ themes table to the structures},
  fonttitle=\small\bfseries, coltitle=black, colbacktitle=black!10]
\small
Four roots and three themes share 126 documents. For each cell,
Table~\ref{tab:box} gives the documents $n_{ab}$ and the adjusted standardised
residual $r_{ab}$ of Eq.~\eqref{eq:residual}; links are in bold.

\begin{center}
\captionof{table}{Roots $\times$ themes table of Box~1: documents $n_{ab}$,
adjusted standardised residuals $r_{ab}$ and the structures they give. Links in
bold.}\label{tab:box}
\begin{tabular}{@{}lrrrrrrl@{}}
\toprule
 & \multicolumn{2}{c}{$T_1$} & \multicolumn{2}{c}{$T_2$} & \multicolumn{2}{c}{$T_3$} & Root \\
\cmidrule(lr){2-3}\cmidrule(lr){4-5}\cmidrule(lr){6-7}
 & $n_{ab}$ & $r_{ab}$ & $n_{ab}$ & $r_{ab}$ & $n_{ab}$ & $r_{ab}$ & \\
\midrule
$R_1$ & \textbf{30} & $\mathbf{4.2}$ & \textbf{25} & $\mathbf{3.3}$ & 5 & $-7.0$ & branching \\
$R_2$ & 4 & $-2.5$ & 6 & $-1.1$ & \textbf{20} & $\mathbf{3.3}$ & consolidated \\
$R_3$ & 6 & $-1.8$ & 4 & $-2.2$ & \textbf{22} & $\mathbf{3.8}$ & consolidated \\
$R_4$ & 0 & $-1.4$ & 0 & $-1.3$ & 4 & $2.5$ & dispersed \\
\midrule
Theme & \multicolumn{2}{c}{consolidated} & \multicolumn{2}{c}{consolidated} & \multicolumn{2}{c}{convergent} & \\
\bottomrule
\end{tabular}

\end{center}

$R_1$ contributes to $T_1$ and $T_2$ more documents than their sizes would
predict: it is a \emph{branching root}, one intellectual base that has grown
into two topics. $T_3$ draws on two roots, $R_2$ and $R_3$: it is a
\emph{convergent theme}, a topic reached from two bases. $R_2$, $R_3$, $T_1$ and
$T_2$ have one link each and are \emph{consolidated}. $R_4$ has a residual of
2.5 with $T_3$, but on four documents, fewer than $n_{\min}$, so $R_4$
has no link and is \emph{dispersed}. The same rule is read along the rows for
the roots and down the columns for the themes, so the two classifications
cannot disagree about which links exist.

\smallskip\noindent\textit{Note.} In this example, the thresholds of the link rule are a minimum of
$n_{\min}=5$ documents and $z=2$.
\end{tcolorbox}

\subsection{Cohesion and lift}\label{sec:cohesion}

How close the documents of a cluster are, and how close two roots are, is
measured as a mean cosine relative to the mean cosine of two random documents.
Let $\mathbf{u}_i=\mathbf{x}_i/\lVert\mathbf{x}_i\rVert$ be the unit row of
$d_i$ in a layer and $\mathbf{S}_C=\sum_{i\in C}\mathbf{u}_i$ the sum over a
set of documents $C$.

\begin{proposition}\label{prop:cohesion}
For two disjoint sets of documents $G$ and $H$,
\begin{equation}\label{eq:between}
  \bar s(G,H) = \frac{1}{|G|\,|H|}\sum_{i\in G}\sum_{j\in H} s_{ij}
  = \frac{\mathbf{S}_G^{\top}\mathbf{S}_H}{|G|\,|H|},
\end{equation}
and, for a set $G$ of non-zero rows,
\begin{equation}\label{eq:within}
  \bar s(G) = \frac{1}{|G|(|G|-1)}\sum_{i\in G}\sum_{j\in G,\,j\ne i} s_{ij}
  = \frac{\lVert\mathbf{S}_G\rVert^2-|G|}{|G|(|G|-1)}.
\end{equation}
\end{proposition}

\begin{proof}
With the cosine of Eq.~\eqref{eq:cosine}, $s_{ij}=\mathbf{u}_i^{\top}\mathbf{u}_j$.
Then
\[
  \textstyle\sum_{i\in G}\sum_{j\in H}\mathbf{u}_i^{\top}\mathbf{u}_j
  =(\sum_{i\in G}\mathbf{u}_i)^{\top}(\sum_{j\in H}\mathbf{u}_j)
  =\mathbf{S}_G^{\top}\mathbf{S}_H .
\]
For the second equality,
$\lVert\mathbf{S}_G\rVert^2=\sum_{i\in G}\sum_{j\in G}\mathbf{u}_i^{\top}\mathbf{u}_j$
includes the $|G|$ terms $\mathbf{u}_i^{\top}\mathbf{u}_i=1$, which are
subtracted.
\end{proof}

With the proposition the indices are exact and no pair is enumerated. With $\bar s(\mathcal{D})$ the mean cosine of two random documents of the
collection, the \emph{cohesion} ($\kappa$) of a cluster $C$ and the \emph{lift} ($\lambda$)
between two roots are
\begin{equation}\label{eq:lift}
  \kappa(C) = \frac{\bar s(C)}{\bar s(\mathcal{D})}, \qquad
  \lambda(R_a,R_b) = \frac{\bar s(R_a,R_b)}{\bar s(\mathcal{D})},
\end{equation}
computed in both layers. A lift of 1 means that a document of $R_a$ and one of
$R_b$ are as close as two random documents. Cohesion and lift use all
similarities, not only those of the neighbour graphs, and always the cosine.

\subsection{Pairs of roots, static and over time}\label{sec:pairs}\label{sec:evolution}

The links relate roots to themes; the framework also relates roots to each
other. Two roots are distinct groups in the references layer by construction,
but they may be more or less distant in it, and they may or may not write about
the same subjects. Reading every pair of roots in both layers at once separates
four relations between intellectual bases. Two roots close in both layers share
part of their base and apply it to the same subjects. Two roots close in
references only share part of their base but apply it to different subjects.
Two roots close in topics only are the cluster-level counterpart of convergence:
different bases that have reached the same subjects. Two roots close in neither
are unrelated. The lifts of Eq.~\eqref{eq:lift} measure the two proximities on
the same scale, so each pair of roots is described by
$(\lambda^{\mathrm{ref}},\lambda^{\mathrm{top}})$, and the value 1, the
proximity of two random documents, separates close from not close in each layer.
When two roots close in topics are both linked to the same theme, that theme
names what brings them together.

These relations are not fixed. Bases that start apart may come to write about
the same subjects, and bases that once shared their topics may specialise. The
evolutionary view follows the same two lifts over time. Longitudinal science
mapping usually clusters each period separately and then links the clusters of
successive periods \citep{aria2026b}. Here documents belong to one period only,
so clusters found separately in each period could not be matched through their
members. The roots are therefore defined once, on the whole
collection, and in each period the lifts are computed on the documents of that
period, relative to two random documents of the same period. A root is
followed in the periods in which it has enough documents. For each pair of roots
seen in at least two periods, the slopes of $\log_2\lambda^{\mathrm{ref}}$ and
$\log_2\lambda^{\mathrm{top}}$ on the period index summarise the trend. With a
threshold $m$ on the slopes, a pair of roots is \emph{converging} when it gets
closer in topics but not in references, \emph{diverging} when it gets farther in
topics but not in references, \emph{consolidating} when it gets closer in both
and \emph{drifting apart} when it gets farther in both; when its topic proximity
does not change, it is \emph{closer in references} or \emph{farther in
references} if its reference proximity does, and \emph{stable} otherwise. Pairs
that move from one of the four relations to another are the ones that change
relation.

Two displays are proposed for this reasoning. The \emph{plane of the pairs of
roots} shows the static view: each pair is a point at
$(\lambda^{\mathrm{ref}},\lambda^{\mathrm{top}})$ on logarithmic axes, the lines
$\lambda=1$ split the plane into the four areas \emph{close in both},
\emph{close in topics only}, \emph{close in references only} and \emph{close in
neither}, and pairs close in topics whose roots are linked to the same theme are
labelled with it. The
\emph{trajectory plot} shows the evolutionary view on the same plane: each pair
that changes area is drawn as a path from its first period to its last, with the
line type giving its trend. Figure~\ref{fig:pairs-example} sketches both
displays on a stylised example.

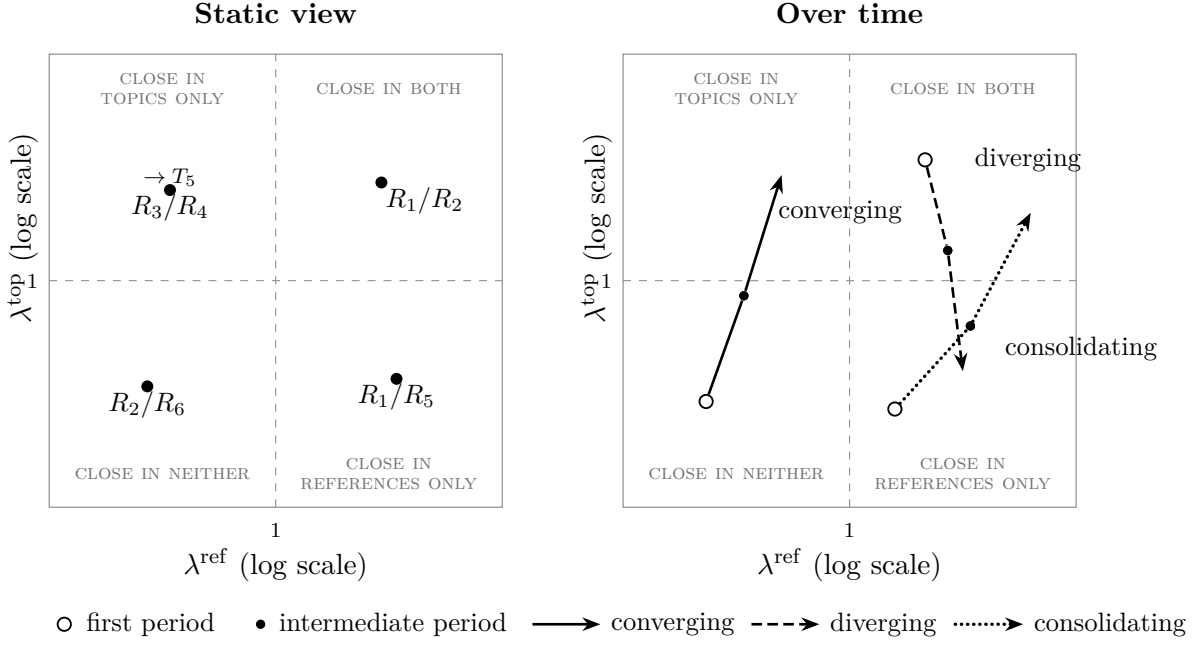
\begin{figure}[!tbp]
  \centering
  {\linespread{1}\selectfont\resizebox{\textwidth}{!}{\begin{tikzpicture}[
  font=\small,
  area/.style={font=\scriptsize\scshape, text=black!55, align=center},
  pt/.style={circle, fill=black, inner sep=0pt, minimum size=4.5pt},
  start/.style={circle, draw=black, fill=white, line width=0.7pt, inner sep=0pt, minimum size=5pt},
  mid/.style={circle, fill=black, inner sep=0pt, minimum size=3.5pt},
  lab/.style={font=\footnotesize, inner sep=1.5pt},
  per/.style={font=\scriptsize, text=black!60, inner sep=1pt},
  conv/.style={line width=1pt, -{Stealth[length=2.4mm]}},
  divg/.style={line width=1pt, dash pattern=on 4pt off 2pt, -{Stealth[length=2.4mm]}},
  cons/.style={line width=1.3pt, dash pattern=on 0.1pt off 2.2pt, line cap=round, -{Stealth[length=2.4mm]}},
  panel/.style={font=\small\bfseries}
]
\newcommand{\plane}{%
  \draw[black!50] (0,0) rectangle (6,6);
  \draw[black!45, dashed] (3,0) -- (3,6);
  \draw[black!45, dashed] (0,3) -- (6,3);
  \node[area] at (1.5,5.55) {close in\\[-2pt]topics only};
  \node[area] at (4.5,5.55) {close in both};
  \node[area] at (1.5,0.45) {close in neither};
  \node[area] at (4.5,0.45) {close in\\[-2pt]references only};
  \node[font=\scriptsize] at (3,-0.3) {1};
  \node[font=\scriptsize, left] at (0,3) {1};
  \node[below=4mm] at (3,0) {$\lambda^{\mathrm{ref}}$ (log scale)};
  \node[rotate=90, above=7mm] at (0,3) {$\lambda^{\mathrm{top}}$ (log scale)};
}

\begin{scope}
  \plane
  \node[panel] at (3,6.55) {Static view};
  \node[pt] (a) at (4.4,4.3) {}; \node[lab, below right] at (a) {$R_1/R_2$};
  \node[pt] (b) at (1.6,4.2) {}; \node[lab, below] at (b) {$R_3/R_4$};
  \node[lab, above, font=\scriptsize\itshape] at (b) {$\to T_5$};
  \node[pt] (c) at (4.6,1.7) {}; \node[lab, below] at (c) {$R_1/R_5$};
  \node[pt] (d) at (1.3,1.6) {}; \node[lab, below] at (d) {$R_2/R_6$};
\end{scope}

\begin{scope}[xshift=7.6cm]
  \plane
  \node[panel] at (3,6.55) {Over time};
  \node[start] (c1) at (1.1,1.4) {}; \node[mid] (c2) at (1.6,2.8) {};
  \draw[conv] (c1) -- (c2) -- (2.1,4.4); \node[mid] at (1.6,2.8) {};
  \node[lab, right] at (2.0,3.9) {converging};
  \node[start] (d1) at (4.0,4.6) {}; \node[mid] (d2) at (4.3,3.4) {};
  \draw[divg] (d1) -- (d2) -- (4.5,1.8); \node[mid] at (4.3,3.4) {};
  \node[lab, right] at (4.6,4.6) {diverging};
  \node[start] (s1) at (3.6,1.3) {}; \node[mid] (s2) at (4.6,2.4) {};
  \draw[cons] (s1) -- (s2) -- (5.4,3.9); \node[mid] at (4.6,2.4) {};
  \node[lab, right] at (5.0,2.1) {consolidating};
\end{scope}

\begin{scope}[yshift=-1.55cm, xshift=0.2cm]
  \node[start] at (0,0) {}; \node[right, font=\footnotesize] at (0.15,0) {first period};
  \node[mid] at (2.6,0) {}; \node[right, font=\footnotesize] at (2.7,0) {intermediate period};
  \draw[conv] (6.2,0) -- (7.1,0); \node[right, font=\footnotesize] at (7.1,0) {converging};
  \draw[divg] (9.1,0) -- (10.0,0); \node[right, font=\footnotesize] at (10.0,0) {diverging};
  \draw[cons] (11.8,0) -- (12.7,0); \node[right, font=\footnotesize] at (12.7,0) {consolidating};
\end{scope}
\end{tikzpicture}}}
  \caption{The two displays of the pairs of roots on a stylised example.
  Left, the static view: each pair of roots is a point at
  $(\lambda^{\mathrm{ref}},\lambda^{\mathrm{top}})$, and the lines at 1 split
  the four relations; $R_3$ and $R_4$, close in topics only, are both linked to
  theme $T_5$, which labels the pair. Right, the view over time: each pair moves
  from its first period (hollow point) to its last (arrow), and the line type
  gives its trend. Positions are illustrative.}
  \label{fig:pairs-example}
\end{figure}

\subsection{Robustness to the number of neighbours}\label{sec:robustness}

The number of neighbours $k$ sets how dense the two graphs are, and so how
fine roots and themes are. As $k$ grows, clusters tend to merge rather than
reshuffle, so that $k$ acts largely as
the resolution does in publication-level classifications, where the granularity
has to be chosen against an external criterion \citep{sjogarde2020}; partitions found at different $k$ are
nonetheless different. Rather than fixing
a single value, the framework reports which of its conclusions persist.

Only the neighbour graphs depend on $k$; the layers, the weights and the
documents do not. Roots and themes are therefore found again at other values of
$k$, below and above the one chosen, from the same layers, with all other
choices unchanged; the partitions are those of a new analysis with that $k$, at
the cost of the clustering only. A link between $R_a$ and $T_b$
\emph{persists} at another $k$ when most of its $n_{ab}$ documents fall in a
single cell of the new roots $\times$ themes table that is itself a link; a
root or a theme keeps its structure when most of its documents fall in a new
root or theme with the same structure. The \emph{persistence} of a link, root or
theme is the share of the other values of $k$ at which this happens. The
interpretation rests on the links and the structures that persist at every $k$.

\section{Empirical application}\label{sec:application}

The framework was applied to three collections from three bibliographic
sources, chosen to differ in field, in size and in the way cited references are
recorded. The \emph{management collection} is the one used throughout
\citet{aria2026book}: \BookDocs{} original research articles in English,
indexed by Web of Science between 1985 and 2025, that use bibliometric or
scientometric methods in management and business. The \emph{explainable AI
collection} contains \XaiDocs{} articles from Scopus whose titles mention
``explainable artificial intelligence'', ``explainable AI'' or ``XAI''
(2017--2025).
The \emph{disinformation collection} contains \DisDocs{} journal articles from
OpenAlex \citep{priem2022} with the term in the title (2000--2025, duplicates
removed). The OpenAlex data, the identifiers of the Scopus records and the
scripts of all analyses are deposited on Zenodo \citep{aria2026data}. This section presents the
management collection, a field whose literature is familiar to most readers and
easy to interpret. The other two collections are reported in full in the
Supplementary Material (Section~S5), with their data, analysis, results and
interpretation; the findings they share with the management collection are
taken up in the conclusions.

\paragraph{Analytical choices}
A document enters the analysis when it has at least one cited reference and at
least one term. The topic layer is built from the author keywords and the
Keywords Plus, since these are available for almost all documents; titles and
abstracts would be used instead if keywords were missing for a large share of
the collection. Terms used by a single document are removed,
and so are the placeholders that Web of Science writes for unreadable
references (``NO TITLE CAPTURED'', ``ANONYMOUS''), which would couple unrelated
documents. The keywords were cleaned with lists of synonyms and terms to remove:
spelling variants, singular and plural forms and acronyms were merged into one
term, and the terms of the search query, generic words and indexing boilerplate
were removed; the lists are available with the replication material. The
analysis uses $k=10$ nearest neighbours, $B=99$ randomised matrices and
$\alpha=0.05$ for the null model, a random sample of $10^5$ pairs for the
association, consensus partitions of 20 runs, and the link rule with $z=2$ and
$n_{\min}=5$. Clusters of fewer than five documents are not kept, and their
documents are reported as being in no root or in no theme. Persistence is
measured at $k=5$ and $k=20$. In the evolution, a root is followed in a period
when it has at least five documents in it, and the threshold on the slopes is
$m=0.1$. These are the choices of a first analysis, not part of the method;
Section~\ref{sec:sensitivity} varies them. All the analyses were run with the development
version of \textit{bibliometrix} \citep{aria2017}, available on GitHub, in
which the framework is implemented (Supplementary Material, Section~S7).

\subsection{The management collection}\label{sec:book}

Of the \BookDocs{} documents, \BookAnalysed{} have both cited references and
keywords and enter the analysis; keywords cover \BookTopicCoverage{} of the
documents with references, so the topic layer is built from them. The two
layers carry largely complementary information. \BookPairsBoth{}\% of the
neighbour pairs are neighbours in both layers, and on random pairs the
correlation between the two similarities is \BookRhoRandom{}: positive, so that
management studies that cite the same works tend to share keywords, but far
from redundant. Citing the same works and writing about the same subjects are
related but distinct ways in which two documents are close, and the two layers
keep them apart.

The references layer gives \BookRoots{} roots, \BookBranching{} of them branching,
and the topic layer \BookThemes{} themes, \BookConvergent{} of them convergent,
connected by \BookLinks{} links (Table~\ref{tab:book-roots} and
Figs.~\ref{fig:book-matrix} and~\ref{fig:book-links}; the themes are listed in
the Supplementary Material, Section~S3). Of these links, \BookLinksPersist{}
persist at every $k$; so do \BookBranchingPersist{} of the branching roots and
\BookConvergentPersist{} of the convergent themes. The interpretation below rests
on these. The roots $\times$ themes table (Fig.~\ref{fig:book-matrix}) shows
the whole relation at a glance: most roots concentrate their documents in a few
themes, while the largest root, R1, has documents in almost all of them.
Table~\ref{tab:book-profiles} gives, for the \BookProfiledRoots{} roots whose
structure persists, the keywords that characterise their documents and their strongest
references, on which the reading that follows is based.

\begin{figure}[!tbp]
  \centering
  \includegraphics[width=\textwidth]{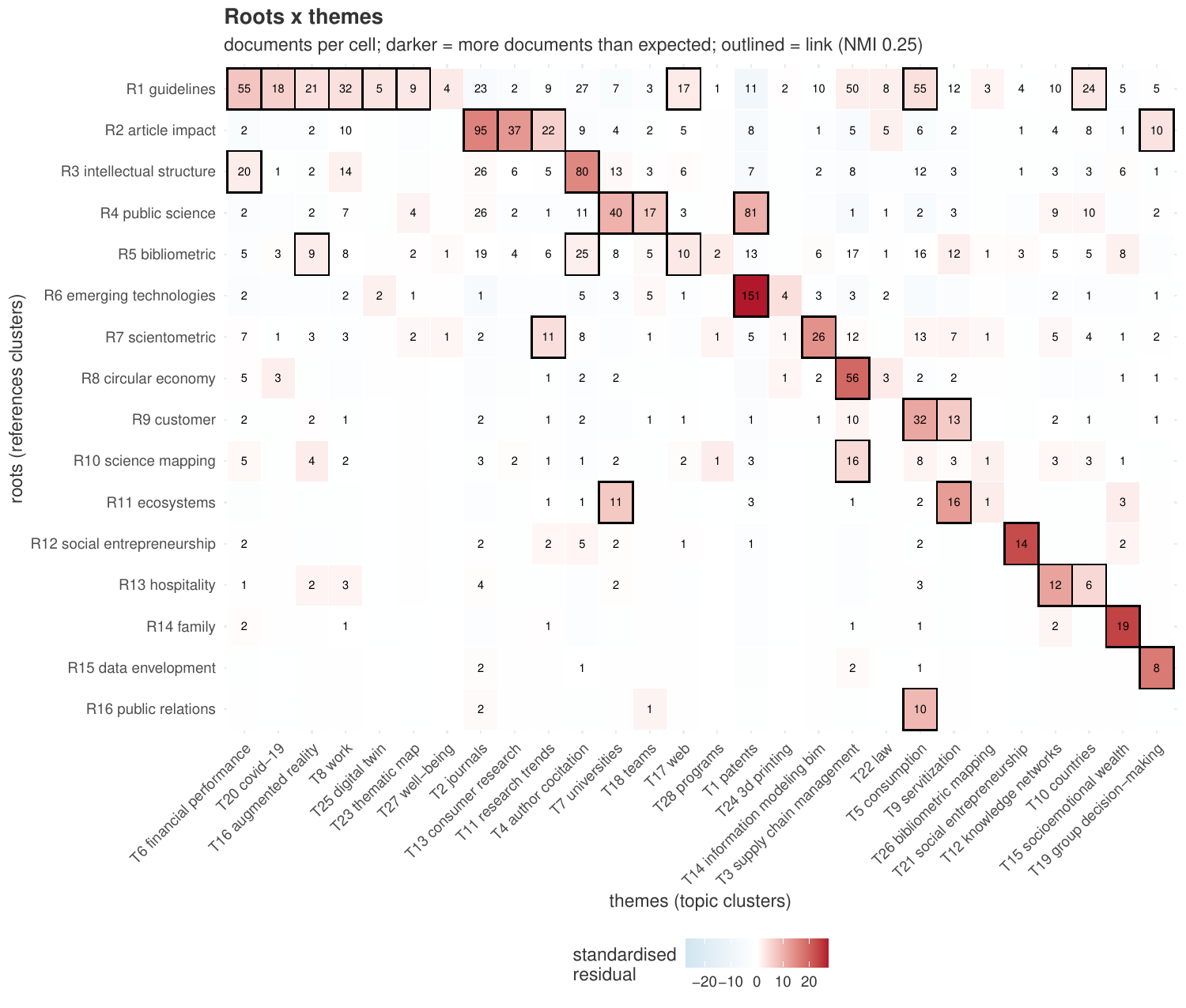}
  \caption{Roots $\times$ themes table of the management collection. Each cell
  gives the documents $n_{ab}$ of root $R_a$ in theme $T_b$; the shading is the
  adjusted standardised residual $r_{ab}$, darker when the root contributes more
  documents than expected, and outlined cells are links.}
  \label{fig:book-matrix}
\end{figure}

\begin{table}[!tbp]
  \centering
  \caption{Roots of the management collection. Name: the most characteristic terms of
  the titles of the root's strong references (Section~\ref{sec:clusters}).
  Persists: the structure of the root is the same at $k=5$ and $k=20$.
  Linked themes: themes with $r_{ab}>2$ and $n_{ab}\ge5$, ordered by documents;
  an asterisk marks a link that does not persist at every $k$.}
  \label{tab:book-roots}
  \scriptsize
  \resizebox{\textwidth}{!}{\begin{tabular}{@{}llrllp{3.2cm}@{}}
\toprule
Root & Name (strong references) & Docs & Structure & Persists & Linked themes \\
\midrule
R1 & guidelines; literature reviews; conduct & 434 & branching & yes & T5$^{*}$, T6$^{*}$, T8, T10, T16, T20, T17$^{*}$, T23, T25$^{*}$ \\
R2 & article impact; citation; operations management & 240 & branching & no & T2, T13, T11, T19$^{*}$ \\
R3 & intellectual structure; author co-citation; strategic management & 226 & branching & yes & T4, T6$^{*}$ \\
R4 & public science; collaboration; geographic localization & 225 & branching & yes & T1, T7, T18 \\
R5 & bibliometric; bibliographic coupling; scientific papers & 194 & branching & no & T4, T17$^{*}$, T16$^{*}$ \\
R6 & emerging technologies; patent; forecasting & 190 & consolidated & no & T1 \\
R7 & scientometric; emerging trends; building information & 118 & branching & yes & T14, T11$^{*}$ \\
R8 & circular economy; supply chain; chain management & 81 & consolidated & no & T3 \\
R9 & customer; experience; engagement & 73 & branching & no & T5, T9$^{*}$ \\
R10 & science mapping; software tool; scimat & 63 & consolidated & no & T3$^{*}$ \\
R11 & ecosystems; entrepreneurial; two-sided markets & 39 & branching & no & T9, T7 \\
R12 & social entrepreneurship; hybrid; social enterprises & 33 & consolidated & no & T21 \\
R13 & hospitality; tourism; strategic management & 33 & branching & no & T12$^{*}$, T10 \\
R14 & family; firms; socioemotional wealth & 27 & consolidated & yes & T15 \\
R15 & data envelopment; efficiency; health care & 14 & consolidated & yes & T19 \\
R16 & public relations; knowledge construction; dialogic & 13 & consolidated & yes & T5 \\
\bottomrule
\end{tabular}
}
\end{table}

\IfFileExists{tab_book_root_profiles.tex}{%
\begin{table}[!tbp]
  \centering
  \caption{Profiles of the roots of the management collection whose structure
  persists, i.e.\ is the same at $k=5$ and $k=20$ (names and sizes in
  Table~\ref{tab:book-roots}). Characteristic keywords:
  the five keywords with the highest $n_{af}^2/n_f$, where $n_{af}$ is the number
  of documents of the root with keyword $f$ and $n_f$ the number in the
  collection. Strong references: the three references with the highest strength
  (Section~\ref{sec:clusters}), in alphabetical order; first author, year and
  source as recorded in the cited reference.}
  \label{tab:book-profiles}
  \scriptsize
  \begin{tabular}{@{}l l p{0.33\linewidth} p{0.42\linewidth}@{}}
\toprule
Root & Structure & Characteristic keywords & Strong references \\
\midrule
R1 & branching & \raggedright systematic literature review; artificial intelligence; behavior; sustainability; business & \raggedright Aria (2017), J Informetr; Donthu (2021), J Bus Res; Zupic (2015), Organ Res Methods \tabularnewline
\addlinespace
R3 & branching & \raggedright author cocitation; intellectual structure; co-citation analysis; resource-based view; firms & \raggedright Barney (1991), J Manage; McCain (1990), J Am Soc Inform Sci; Ramos-Rodríguez (2004), Strategic Manage J \tabularnewline
\addlinespace
R4 & branching & \raggedright indicators; innovation; universities; technology transfer; patents & \raggedright Griliches (1990), J Econ Lit; Jaffe (1993), Q J Econ; Katz (1997), Res Policy \tabularnewline
\addlinespace
R7 & branching & \raggedright emerging trends; information modeling bim; scientometric analysis; construction industry; topic modeling & \raggedright Chen (2006), J Am Soc Inf Sci Tec; Hosseini (2018), Automat Constr; Kleinberg (2003), Data Min Knowl Disc \tabularnewline
\addlinespace
R14 & consolidated & \raggedright socioemotional wealth; family firms; family business; ownership; generational involvement & \raggedright Berrone (2012), Fam Bus Rev; Chua (1999), Entrep Theory Pract; Gómez-Mejía (2007), Admin Sci Quart \tabularnewline
\addlinespace
R15 & consolidated & \raggedright technical efficiency; data envelopment analysis; stochastic frontier analysis; bootstrap; literature survey & \raggedright Banker (1984), Manage Sci; Charnes (1978), Eur J Oper Res; Farrell (1957), J R Stat Soc Ser A-G \tabularnewline
\addlinespace
R16 & consolidated & \raggedright knowledge construction; public relations; public-relations research; invisible college; engagement & \raggedright Morehouse (2018), J Public Relat Res; Pasadeos (1992), Journal of Public Relations Research; Pasadeos (2010), J Public Relat Res \tabularnewline
\bottomrule
\end{tabular}

\end{table}}{}

\begin{figure}[!tbp]
  \centering
  \includegraphics[width=\textwidth]{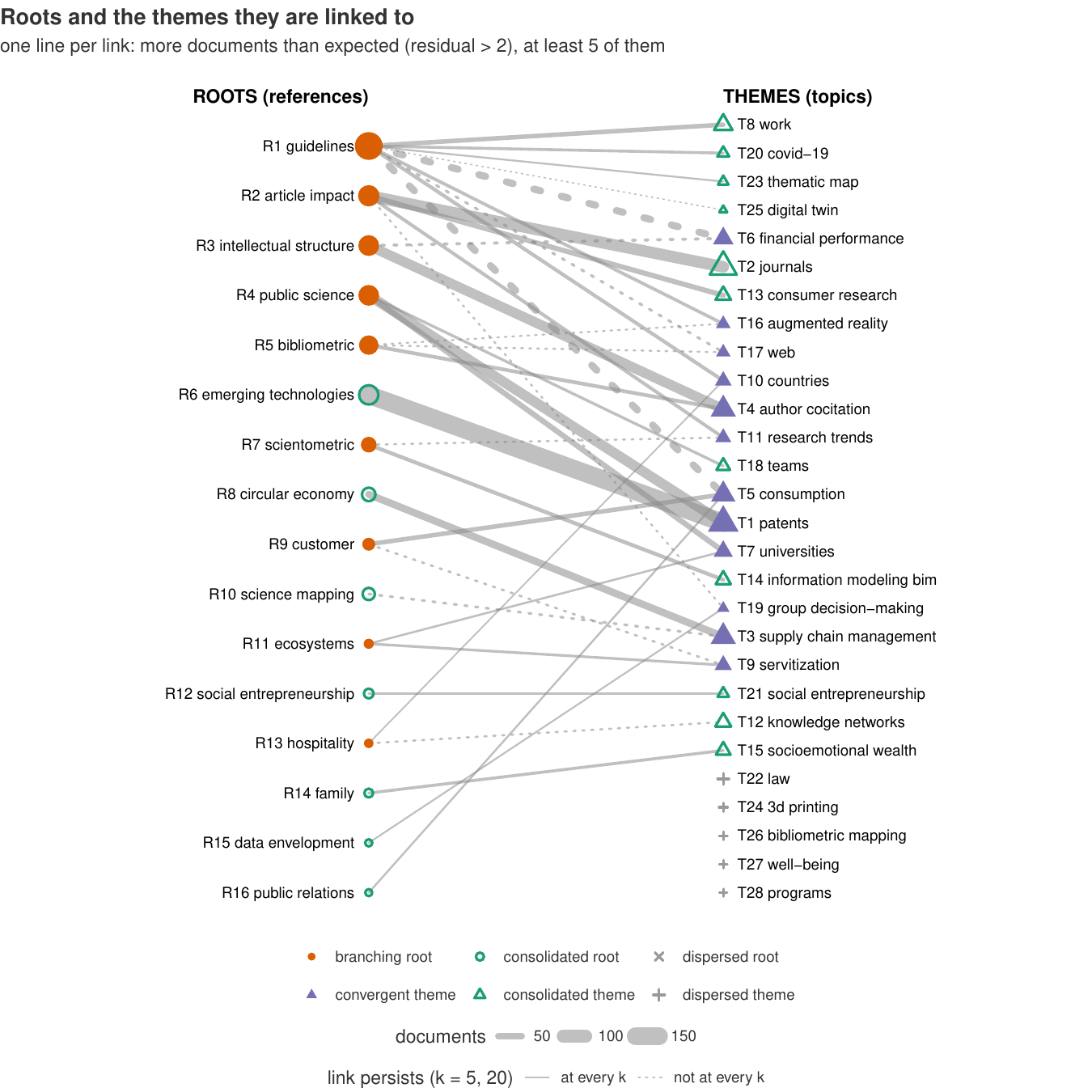}
  \caption{Roots and themes of the management collection and their links. Roots on the
  left, themes on the right; the width of a line is the number of documents of
  the link, and dashed lines are links that do not persist at every $k$. Roots
  are circles and themes triangles: filled when branching or convergent, open when
  consolidated, a cross when dispersed.}
  \label{fig:book-links}
\end{figure}

The four persistent branching roots show two different ways in which an
intellectual base grows into several topics. Three of them are substantive
bases. R4, named ``public science; collaboration; geographic localization'' after
the economics of innovation and science it draws on, feeds three themes, on
patents and technology indicators (T1), on universities and academic
entrepreneurship (T7), and on research teams and interdisciplinarity (T18), all
three links persisting. R3, ``intellectual structure; author co-citation;
strategic management'', is the base of the co-citation studies of strategic
management and feeds both the theme on author co-citation and dynamic
capabilities (T4) and, less firmly, the one on corporate social responsibility
and financial performance (T6). R7, ``scientometric; emerging trends'', moves
between research trends and knowledge management (T11) and building information
modelling (T14). The fourth, R1, ``guidelines; literature reviews; conduct'', is
of another kind: the largest root (\BookRootOneSize{} documents) is the methodological base of
bibliometric reviews, and its documents spread over nine themes, from work and
ethical leadership to tourism, virtual reality and nursing during the pandemic.
Five of its nine links persist. A root of this kind is a shared method rather
than a shared subject, and the framework identifies it as such, since its
documents are close in references and far apart in topics.

Convergence appears where different bases meet on a subject. The theme on
patents, indicators and emerging technologies (T1) draws on the economics of
public science (R4) and on technology forecasting (R6, ``emerging technologies;
patent; forecasting''); the theme on author co-citation and dynamic capabilities
(T4) on the co-citation tradition (R3) and on the literature on bibliographic
coupling (R5). Both convergences persist at every $k$, links included.

\begin{figure}[!tbp]
  \centering
  \includegraphics[width=\textwidth]{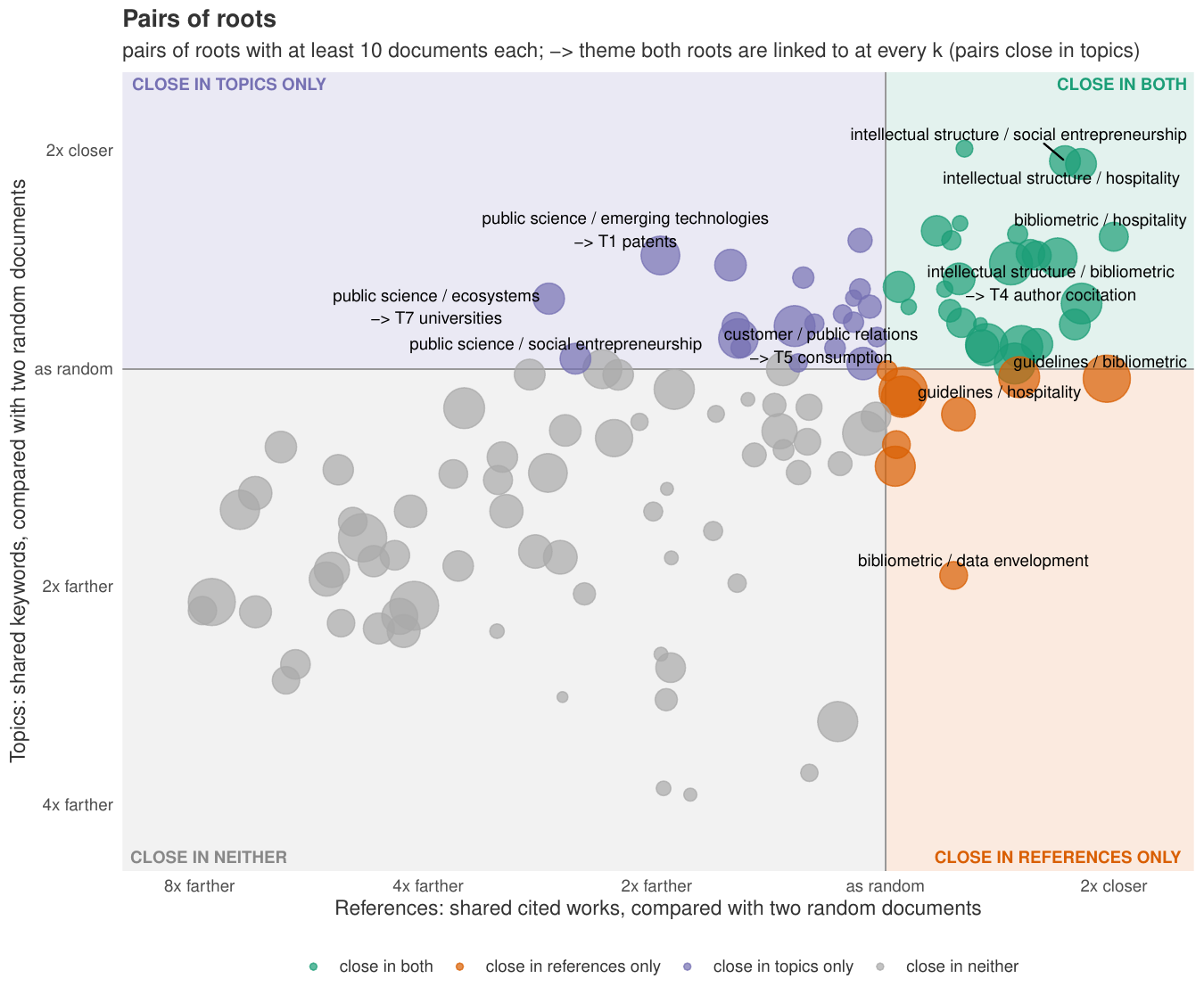}
  \caption{Pairs of roots of the management collection on the plane of their lifts in
  references (horizontal) and topics (vertical), on a logarithmic scale; the lines
  at 1 (\emph{as random}) split the four areas. Named pairs: the three furthest from \emph{as random} in each area other than
  \emph{close in neither}, and the pairs close in topics whose two roots are
  linked to the same theme by links that persist at every $k$, labelled with
  that theme.}
  \label{fig:book-plane}
\end{figure}

On the plane of the pairs of roots (Fig.~\ref{fig:book-plane}), \BookPairsNeither{}
of the \BookPairs{} pairs are close in neither layer, \BookPairsBothAreas{} close
in both, \BookPairsTopOnly{} close in topics only and \BookPairsRefOnly{} close in
references only.

\subsection{Evolution}\label{sec:book-evolution}

The evolution uses the periods of the thematic evolution of the same collection
in \citet{aria2026book}, 1985--2004, 2005--2010, 2011--2015 and 2016--2025, with
the last period split at 2020 because more than three quarters of the documents
were published after 2015 (Fig.~\ref{fig:book-production}). A root is followed in
a period when it has at least five documents in it; many roots appear only in
the last two periods, so most pairs of roots are followed in two periods only.
Of the \EvoPairs{} pairs followed in at least two periods, \EvoPairsThree{} are
followed in three or more, and \EvoPairsMove{} change area of the plane.

\begin{figure}[!tbp]
  \centering
  \includegraphics[width=0.85\textwidth]{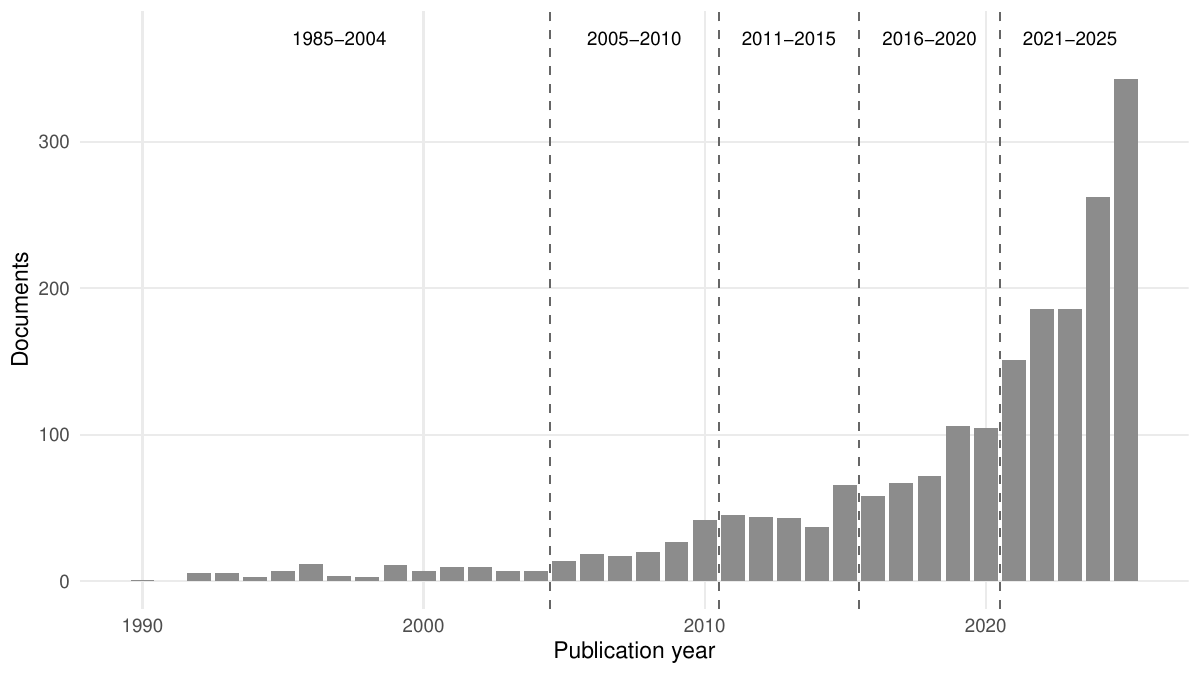}
  \caption{Annual production of the documents analysed in the management collection,
  with the periods of the evolution (dashed lines at the cutting years).}
  \label{fig:book-production}
\end{figure}

\begin{figure}[!tbp]
  \centering
  \includegraphics[width=\textwidth]{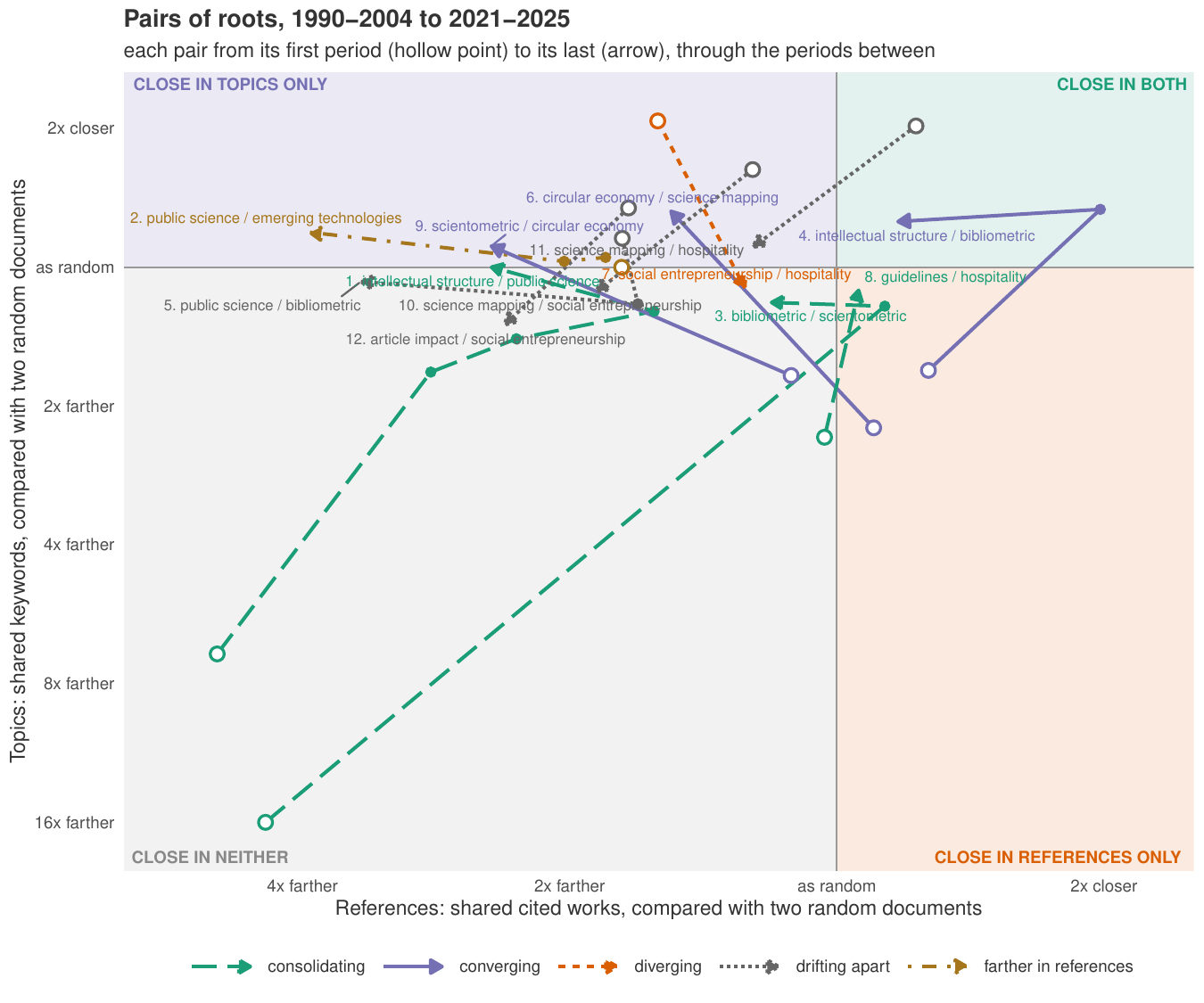}
  \caption{Trajectories of the pairs of roots that change area, from their first
  period (hollow point) to their last (arrow). The line type gives the trend of
  the pair (Section~\ref{sec:evolution}).}
  \label{fig:book-trajectory}
\end{figure}

Over these pairs, \EvoConverging{} are converging, \EvoConsolidating{}
consolidating, \EvoDriftingApart{} drifting apart and \EvoDiverging{} diverging;
\EvoFartherinReferences{} grow apart in references with stable topics,
\EvoCloserinReferences{} closer, and \EvoStable{} do not change. Trends estimated
on two periods are fragile, and the reading of Fig.~\ref{fig:book-trajectory}
privileges the pairs followed longer. The economics of public science (R4) and
the co-citation tradition of strategic management (R3), followed over all five
periods, move from close in neither layer to close in topics only. The two bases
remain distinct in what they cite and come to write about the same subjects.

\subsection{Sensitivity and stability}\label{sec:sensitivity}

Table~\ref{tab:sensitivity} varies one choice at a time on the management
collection. The number of neighbours is the parameter with the largest effect:
the adjusted Rand index with the baseline partition is \SensKFiveRoots{} for the
roots and \SensKFiveThemes{} for the themes at $k=5$, and \SensKTwentyRoots{} and
\SensKTwentyThemes{} at $k=20$. Most of this effect is a change of granularity: as
$k$ grows, clusters mostly merge rather than reshuffle (Supplementary Material,
Section~S4, which also compares alternative constructions of the graph). Restricting the
analysis to a subset of the most cited documents (i.e., 1,000 documents) changes the roots substantially
(\SensTopRoots{}), because it removes the recent documents, which are the least
cited; so does weighting references by idf (\SensIdfRoots{}). Titles and abstracts give themes that agree little
with those from keywords (\SensTiabThemes{}), which is why the source of topics is
reported with the results. Association strength makes the topic layer degenerate
(\SensAssocThemes{} themes, NMI \SensAssocNMI{}). Without keyword cleaning the
themes change too (\SensNoCleanThemes{}).
The threshold on the documents of a link and the minimum size of a cluster
change only the links and the smallest clusters.

\begin{table}[!tbp]
  \centering
  \caption{Sensitivity of the results on the management collection to one choice at a
  time. Roots (br.): roots and branching roots; Themes (conv.): themes and
  convergent themes. ARI with the baseline: adjusted Rand index between the roots,
  or themes, of the variant and those of the baseline analysis, on the documents
  clustered in both.}
  \label{tab:sensitivity}
  \small
  \resizebox{\textwidth}{!}{\begin{tabular}{@{}lrrrrrrr@{}}
\toprule
 & & & & & & \multicolumn{2}{c}{ARI with the baseline} \\
\cmidrule(lr){7-8}
Variant & Docs & Roots (br.) & Themes (conv.) & Links & NMI & Roots & Themes \\
\midrule
baseline & 2,004 & 16 (9) & 28 (12) & 36 & 0.25 & 1.00 & 1.00 \\
$k=5$ & 2,004 & 35 (16) & 35 (21) & 72 & 0.29 & 0.54 & 0.28 \\
$k=20$ & 2,004 & 15 (7) & 15 (11) & 29 & 0.27 & 0.74 & 0.39 \\
topics from TI + AB & 2,054 & 19 (11) & 26 (13) & 43 & 0.25 & 0.87 & 0.12 \\
idf weights of references & 2,004 & 21 (16) & 28 (17) & 48 & 0.28 & 0.35 & 1.00 \\
association strength & 2,004 & 14 (7) & 9 (7) & 21 & 0.14 & 0.42 & 0.12 \\
Jaccard & 2,004 & 17 (11) & 26 (12) & 40 & 0.27 & 0.74 & 0.44 \\
percentile instead of null & 2,004 & 16 (9) & 28 (12) & 36 & 0.25 & 1.00 & 1.00 \\
1,000 most cited documents & 993 & 13 (8) & 16 (10) & 26 & 0.27 & 0.44 & 0.35 \\
$n_{\min}=3$ & 2,004 & 16 (12) & 28 (13) & 41 & 0.25 & 1.00 & 1.00 \\
$n_{\min}=10$ & 2,004 & 16 (8) & 28 (9) & 31 & 0.25 & 1.00 & 1.00 \\
minimum cluster size 10 & 2,004 & 16 (9) & 23 (12) & 37 & 0.25 & 1.00 & 1.00 \\
minimum cluster size 20 & 2,004 & 14 (9) & 22 (11) & 34 & 0.25 & 1.00 & 1.00 \\
Leiden algorithm & 2,004 & 16 (9) & 22 (11) & 34 & 0.29 & 0.84 & 0.64 \\
no keyword cleaning & 2,017 & 19 (13) & 27 (13) & 44 & 0.28 & 0.89 & 0.28 \\
\bottomrule
\end{tabular}
}
\end{table}

Consensus clustering is what makes the partitions reproducible. Across five seeds
far apart, single runs of Louvain agree with each other at
\StabSingleRootsMin{}--\StabSingleRootsMax{} for the roots and
\StabSingleThemesMin{}--\StabSingleThemesMax{} for the themes (mean adjusted Rand
index, three collections; Supplementary Material, Section~S5); consensus
partitions agree at
\StabConsRootsMin{}--\StabConsRootsMax{} and \StabConsThemesMin{}--\StabConsThemesMax{}.

\section{Validation on synthetic data}\label{sec:validation}

A real collection does not say which of its roots branch or which of its themes
converge. The framework is therefore also applied to synthetic collections in
which this structure is planted.

\subsection{Design}\label{sec:sim-design}

Each synthetic collection has six groups of documents by references, $R_1$ to
$R_6$, and seven by topics, $T_1$ to $T_7$, crossed as follows: $R_1$, $R_2$
and $R_6$ write mainly on one theme each ($T_1$, $T_2$, $T_6$); $R_3$ writes on
two themes, $T_3$ and $T_4$, and is the planted branching root; $R_4$ and $R_5$
write on the same theme $T_5$, the planted convergent theme; and $T_7$ is a
diffuse theme fed by $R_1$, $R_2$ and $R_6$, which therefore branch by design
too. Every document of a root draws its references from a pool of 80 references
of that root, and every document of a theme draws its terms from a pool of 30
terms of that theme, both with Zipf popularity. A share of the references and
terms of each document, the \emph{noise}, comes instead from global pools of
2{,}000 references and 800 terms shared by all groups, again with Zipf
popularity: classics cited by everyone and generic terms used by everyone.
Documents carry on average 20 references and six terms.

The grid crosses the number of documents per group ($m=40$, 100, 250), the noise
(0.2, 0.4, 0.6) and the balance of the groups, with three replicates, for
\SimRuns{} collections of \SimMinN{} to \SimMaxN{} documents. In the
\emph{balanced} design the groups have the sizes above; in the
\emph{unbalanced} design $R_1$ is twice as large and $R_6$ smaller, the
second branch of $R_3$ is less than half the first, and $R_5$ contributes to
$T_5$ a third of what $R_4$ does, a configuration closer to real fields. All collections are analysed with the
choices of Section~\ref{sec:application}, keywords as topic layer and no
keyword cleaning.

\subsection{Recovery of roots, themes and structures}\label{sec:sim-recovery}

Table~\ref{tab:recovery} reports, for each design and noise level, the
agreement of roots and themes with the planted groups and how often the planted
structures are found. Roots are recovered almost exactly in every design (mean
adjusted Rand index of at least \SimARIRootsMin{}): shared references are a
strong signal even when most of them come from the global pool. Themes are
harder. Their agreement with the planted groups falls from \SimARIThemesHigh{}
at low noise to \SimARIThemesLowBal{} and \SimARIThemesLowUnb{} at the highest
noise, in the balanced and unbalanced designs: six terms per document leave
little room once a large share of them is generic.

\begin{table}[t]
  \centering
  \caption{Recovery of the planted structure (means over the three sizes and
  the three replicates). ARI: adjusted Rand index between the roots, or the
  themes, and the planted groups. Planted structure found: share of the runs in
  which the cluster that best matches the planted group has the planted
  structure. False structures: share of the roots planted with one theme
  ($R_4$, $R_5$) classified as branching, and of the themes planted with one
  root ($T_1$--$T_4$, $T_6$) classified as convergent.}
  \label{tab:recovery}
  \small
  \resizebox{\textwidth}{!}{\begin{tabular}{@{}llrrrrrrr@{}}
\toprule
 & & \multicolumn{2}{c}{ARI with the planted groups} & \multicolumn{3}{c}{Planted structure found} & \multicolumn{2}{c}{False structures} \\
\cmidrule(lr){3-4}\cmidrule(lr){5-7}\cmidrule(lr){8-9}
Groups & Noise & Roots & Themes & $R_3$ branching & $T_5$ convergent & $T_7$ convergent & Branching & Convergence \\
\midrule
balanced & 0.2 & 1.00 & 0.98 & 100\% & 100\% & 100\% & 0\% & 0\% \\
 & 0.4 & 1.00 & 0.90 & 100\% & 100\% & 100\% & 0\% & 0\% \\
 & 0.6 & 0.98 & 0.65 & 100\% & 100\% & 67\% & 0\% & 0\% \\
unbalanced & 0.2 & 1.00 & 0.98 & 100\% & 100\% & 67\% & 0\% & 0\% \\
 & 0.4 & 1.00 & 0.88 & 100\% & 100\% & 78\% & 0\% & 0\% \\
 & 0.6 & 0.97 & 0.54 & 100\% & 78\% & 78\% & 22\% & 0\% \\
\bottomrule
\end{tabular}
}
\end{table}

The branching root $R_3$ is found in every run. The convergent theme $T_5$ is
found in every run except at the highest noise of the unbalanced design
(\SimTFiveUnbHigh{}), where the small contribution of $R_5$ is sometimes lost
with the borders of the theme. The diffuse theme $T_7$ is the hardest structure,
found in at least \SimTSevenMin{} of the runs of each design. No theme planted
with a single root is ever classified as convergent.

False branching appears only at the highest noise of the unbalanced design,
where \SimFalseBranchRoots{} of the \SimFalseBranchOf{} roots planted with a
single theme are classified as branching. In these runs the planted theme $T_5$
splits into a main theme and a small fragment (adjusted Rand index of the themes
between \SimFalseARIThemesMin{} and \SimFalseARIThemesMax{}), and the root linked
to $T_5$ is linked to the fragment as well, by a link of \SimWeakNMin{} to
\SimWeakNMax{} documents with a residual between \SimWeakResMin{} and
\SimWeakResMax{}. False branching is therefore a consequence of an unstable
theme, not of the link rule.

Since it originates in an unstable theme, false branching should not survive a
change in the number of neighbours, and the persistence of
Section~\ref{sec:robustness} should recognise it. To test this on more cases,
the framework was applied to \ValRuns{} further collections generated with new
seeds where false branching occurs, at noise 0.5, 0.6 and 0.7 in both designs.
Among their \ValBranching{} branching roots, \ValFalse{} are false
(\ValFalseShareAll{}). At noise 0.7 a second form appears besides the
fragment: the planted theme splits into two parts that receive links of
comparable size from the same root. Among the
\ValPersistent{} branching roots that persist at every $k$, only
\ValFalsePersistent{} are false (\ValFalseSharePersistent{}), against
\ValFalseShareNotPersistent{} among those that do not persist. Requiring
persistence removes \ValFalseRemoved{} of the false branching roots at the cost
of \ValTrueLost{} of the true ones; in the runs of Table~\ref{tab:recovery}, none
of the \DevPersistent{} persistent branching roots is false. Persistence is thus
the diagnostic of false branching, and the reason why the interpretation of the management
collection (Section~\ref{sec:application}) rests on the branching roots and the
links that persist at every $k$. Two further indicators, the imbalance between
the links of a root and the similarity in topics between its themes, were
assessed against a rule fixed in advance and did not improve on persistence out
of sample (Supplementary Material, Section~S2).

\subsection{Link rules and pairs of documents}\label{sec:sim-rules}

The link rule of Section~\ref{sec:links}, with $z=2$ and $n_{\min}=5$, was
chosen among five candidates
(Table~\ref{tab:rules}). All five recover $R_3$, $T_5$ and $T_7$ equally often
and none produces false convergence; they differ in false branching and in how
many links they keep on real data. Requiring the link to reach a minimum share of both the root and the
theme halves false branching in the unbalanced design (\RuleAndFalseUnb{}
instead of \SimFalseBranchUnb{}), but on the management collection it keeps
\RuleAndLinks{} of the \RuleAdoptedLinks{} links and reduces branching roots
from \RuleAdoptedBranching{} to \RuleAndBranching{} and convergent themes from
\RuleAdoptedConvergent{} to \RuleAndConvergent{}. A threshold on the share of
the root penalises large roots, which can supply most of a theme with less than a
tenth of their own documents. The adopted rule keeps those links, and its false
branching is recognised by persistence. Lowering the minimum count to three
documents, or requiring the share of either side, increases false branching.

\begin{table}[t]
  \centering
  \caption{Candidate link rules on the simulation (unbalanced and balanced
  designs, all noise levels) and on the management collection. $T_7$ found: share of
  the unbalanced runs in which the diffuse theme is classified as convergent.
  False branching and false convergence as in Table~\ref{tab:recovery}. Management
  collection: links found by the rule on the roots $\times$ themes table of
  Section~\ref{sec:book}, and the branching roots and convergent themes they
  give; for the two-rule variant, links counted on both sides.}
  \label{tab:rules}
  \small
  \resizebox{\textwidth}{!}{\begin{tabular}{@{}lrrrrrrr@{}}
\toprule
 & \multicolumn{4}{c}{Simulation} & \multicolumn{3}{c}{Management collection} \\
\cmidrule(lr){2-5}\cmidrule(lr){6-8}
 & $T_7$ found & \multicolumn{2}{c}{False branching} & False & & Branching & Convergent \\
Link rule & (unbal.) & bal. & unbal. & convergence & Links & roots & themes \\
\midrule
$r_{ab}>2$, $n_{ab}\ge 5$ (adopted) & 74\% & 0\% & 7\% & 0\% & 36 & 9 & 12 \\
$r_{ab}>2$, $n_{ab}\ge 3$ & 74\% & 2\% & 9\% & 0\% & 41 & 12 & 13 \\
$r_{ab}>2$, 10\% of the root and of the theme & 74\% & 0\% & 4\% & 0\% & 19 & 5 & 5 \\
$r_{ab}>2$, 10\% of the root or of the theme & 74\% & 6\% & 11\% & 0\% & 42 & 12 & 13 \\
$r_{ab}>2$, 10\% of the root (roots) or of the theme (themes) & 74\% & 0\% & 4\% & 0\% & 19 & 6 & 11 \\
\bottomrule
\end{tabular}
}
\end{table}

\label{sec:sim-pairs}Against the planted groups, the situations of single pairs
(Section~\ref{sec:null}) are unreliable except for consolidation: only
\SimPrecBranchMin{} to \SimPrecBranchMax{} of the pairs labelled branching
share the root and not the theme, and only \SimPrecConvMin{} to
\SimPrecConvMax{} of those labelled convergence share the theme and not the
root. This is why the framework draws its conclusions from roots, themes and
their links.

\section{Discussion}\label{sec:discussion}

The framework rests on one premise: the references and the topics of the same
documents describe a field in two ways that do not coincide, and the points where
they disagree are worth studying. The three collections, the management collection of
Section~\ref{sec:application} and the two reported in the Supplementary Material
(Section~S5), confirm the premise: in all of them the correlation between the two similarities on random pairs is
positive but low, from \XaiRho{} to \BookRhoRandom{}, and the two layers carry largely
complementary information.

\subsection{What the second layer adds}\label{sec:disc-adds}

Taken alone, the references layer gives clusters close to those of ordinary
bibliographic coupling (Supplementary Material, Section~S6). What the framework adds is the
relation between two partitions. A hybrid clustering
\citep{janssens2008,thijs2018,boyack2020} resolves the disagreement between
references and terms inside a combined similarity, where it is no longer
visible. Comparisons of citation-based and text-based maps also report limited
overlap between them \citep{boyack2020,xie2025}: low agreement between the
layers appears to be common, and the question is what it means.

The structures of Section~\ref{sec:links} give three answers. A branching root
is an intellectual base whose documents write about different things. In the
management collection, the largest root (R1) has this form because its documents share
the methodological literature of bibliometric reviews and apply it to unrelated
subjects; a hybrid clustering would split it among the subjects, or keep it and
blur them. A convergent theme is a subject reached from different bases, as
patents and technology indicators are reached from the economics of public
science and from technology forecasting. \citet{braam1991} had observed that
different co-citation clusters may identify the same specialty; convergent themes
turn that observation into a structure, defined by a test and checked for
persistence. Finally, the plane of the pairs of roots shows bases that remain
distinct in what they cite while coming to write about the same subjects. In the
terms of \citet{persson1994}, the roots are research fronts identified through
their intellectual base, and the framework connects them to the topics of the
same documents.

\subsection{Robustness as part of the result}\label{sec:disc-robust}

Clusters found by modularity optimisation depend on chance and on choices that
the data do not fix. Chance is addressed by consensus clustering, which raises
the agreement between seeds from \StabSingleThemesMin{}--\StabSingleThemesMax{}
to \StabConsThemesMin{}--\StabConsThemesMax{} for the themes
(Section~\ref{sec:sensitivity}). The number of neighbours cannot be addressed in
the same way, because different values of $k$ give partitions of different
granularity, and the appropriate granularity can only be chosen against an
external criterion \citep{sjogarde2020}. The framework therefore treats the dependence on $k$ as
information: a link, a root or a theme that persists at $k=5$, 10 and 20 does not
depend on the granularity chosen. This is also the most effective protection
against false branching, which comes from themes that split: among the branching
roots that persist at every $k$ only \ValFalseSharePersistent{} are false,
against \ValFalseShareNotPersistent{} among the others
(Section~\ref{sec:sim-recovery}). We recommend reading the persistent structures
first and treating the others as hypotheses.
Further recommendations for use are given in the Supplementary Material.

How the association between layers is measured matters beyond this framework.
Scholarly networks are often compared after sparsification, by thresholds or by
nearest neighbours \citep[e.g.,][]{yan2012}. The
agreement between two layers should be measured on pairs selected on neither.

\subsection{Limitations}\label{sec:limitations}

First, themes are less stable than roots. With a few keywords per document the
topic layer is sparser and noisier, and single runs of the clustering agree at
\BookAgreeThemes{} for the themes of the management collection against
\BookAgreeRoots{} for the roots. The errors of the framework originate in the
themes: the false branching of the synthetic collections and the convergence that
does not persist in the disinformation collection (Supplementary Material,
Section~S5).

Second, persistence identifies the conclusions that do not depend on $k$ but
does not remove the dependence. Alternative constructions of the graphs, based
on shared neighbours, on the significance of similarities or on several values
of $k$ at once, do not remove it either (Supplementary Material, Section~S4).

Third, the topic layer depends on its source: keywords and the terms of titles
and abstracts give themes that agree little (\SensTiabThemes{}), and the
framework cannot say which is right. Keyword cleaning relies on lists prepared by the
analyst.

Fourth, the framework requires cited references in a usable form. Scopus needs
reference matching \citep{aria2026c}. In OpenAlex the coverage of references is
comparable to that of Web of Science and Scopus on recent publications indexed
by all three \citep{culbert2025}, but references are missing for many records,
more often for lower-impact journals and less prestigious institutions
\citep{zhou2026}, and for publications not indexed in the other two sources,
as shown for African research \citep{alonso2025}. Only
\DisAnalysed{} of the \DisDocs{} documents of the disinformation collection
entered the analysis, and the roots describe the documents whose references are
recorded. Sources that do not export cited references cannot be analysed.

Fifth, roots are defined once on the whole collection. This makes them comparable
across periods, but a root that splits over time remains one root, and in growing
fields most pairs of roots are followed in two periods only (\EvoPairsThree{} of
\EvoPairs{} in the management collection are followed in three or more).

Sixth, the names of the roots are an aid to interpretation, not a result. They
rest on the strong references whose title can be found in the collection, in the
reference strings or in OpenAlex, and on their sources when fewer than five
titles are found. As with any algorithmically constructed topical structure,
whose meaning has to be checked against the substance of the field
\citep{held2022velden}, the analyst should read the strong references of a root
before interpreting it.

Seventh, the synthetic collections plant a few structures with fixed pools of
references and terms and no change over time, so the evolution of pairs of roots
has been assessed on real data only.

\section{Conclusions}\label{sec:conclusions}

This paper has introduced multiplex coupling, a framework that couples the same
documents by their references and by their topics, keeps the two layers separate
and studies how they relate. The agreement between the
layers (RQ1) is measured on random pairs of documents and tested by permutation. In the three
collections analysed it is low, and this is the main rationale for a multilayer
analysis: the two layers carry complementary information, and reading them
together enriches what coupling alone tells about a field. Roots that have grown into several
themes and themes reached from several roots (RQ2) are identified by a single
rule on the roots $\times$ themes table; on synthetic data the planted structures
are recovered, and the false branching that unstable themes produce is
recognised by its lack of persistence across numbers of neighbours. The relations
between roots (RQ3) are placed on the plane of their lifts in references and in
topics and followed over periods, which shows pairs of bases that remain distinct
in what they cite while converging on the same subjects.

\paragraph{Future developments}
Three developments follow from the limitations of Section~\ref{sec:limitations}.
The first is to add further layers on the same documents, such as the cited
sources or direct citation, which would turn the comparison of two partitions
into the comparison of several. The second is to let roots change over
time, finding them in each period and connecting them by measures of
lineage between periods, so that a root that forms or splits
can be observed. The third is to strengthen the topic layer, where most of the
instability lies. The similarity between documents in this layer could be
measured on embeddings of their titles, abstracts, keywords or full texts,
which represent content densely rather than through a few shared terms: two
documents that share no keyword could still be found close, and themes should
become as stable as roots are. Embeddings learnt from text alone would be needed,
since models trained on citation links would bring the references back into the
topic layer.

\section*{Author Contributions}
Massimo Aria: Conceptualization, Methodology, Software, Data curation,
Writing---original draft, Writing---review \& editing. Corrado Cuccurullo:
Conceptualization, Writing---review \& editing. Maria Spano: Formal analysis,
Validation, Writing---review \& editing.

\section*{Competing Interests}
The authors have no competing interests.

\section*{Funding Information}
This research received no external funding.

\section*{Use of Generative AI}
The authors used generative AI to assist in editing the text of
the manuscript, under their direction and on the basis of the method, analyses and
interpretations developed by the authors. All results were produced by the
authors' code and data, and the authors reviewed, revised and take full
responsibility for the content of the paper.

\section*{Data Availability}
The management collection is distributed with \textit{bibliometrix} through
\texttt{loadBookDataset()}. The OpenAlex collections, on disinformation and the
larger one used for the computing times, were downloaded with
\textit{openalexR} \citep{aria2024oa} and are deposited in \textit{bibliometrix}
format on Zenodo \citep{aria2026data}, together with the keyword cleaning lists and
the scripts that produce every number, table and figure of the paper. Records
exported from Scopus cannot be redistributed: for the collection on explainable
artificial intelligence the deposit contains the query and the Scopus
identifiers (EID) and DOIs of the documents, from which the collection can be
rebuilt by anyone with access to Scopus. The framework is implemented in
the development version of \textit{bibliometrix} on GitHub (5.5.0.9000, branch
\texttt{feature/multiplex}).

\bibliographystyle{apacite-qss}
\bibliography{refs}

\begin{thebibliography}{}

\bibitem [\protect \citeauthoryear {%
Ahlgren%
\ \protect \BOthers {.}}{%
Ahlgren%
\ \protect \BOthers {.}}{%
{\protect \APACyear {2020}}%
}]{%
ahlgren2020}
\APACinsertmetastar {%
ahlgren2020}%
\begin{APACrefauthors}%
Ahlgren, P.%
, Chen, Y.%
, Colliander, C.%
\BCBL {}\ \BBA {} van Eck, N\BPBI J.%
\end{APACrefauthors}%
\unskip\
\newblock
\APACrefYearMonthDay{2020}{}{}.
\newblock
{\BBOQ}\APACrefatitle {Enhancing direct citations: A comparison of relatedness
  measures for community detection in a large set of {PubMed} publications}
  {Enhancing direct citations: A comparison of relatedness measures for
  community detection in a large set of {PubMed} publications}.{\BBCQ}
\newblock
\APACjournalVolNumPages{Quantitative Science Studies}{1}{2}{714--729}.
\newblock
\begin{APACrefDOI} \doi{10.1162/qss_a_00027} \end{APACrefDOI}
\PrintBackRefs{\CurrentBib}

\bibitem [\protect \citeauthoryear {%
Ahlgren%
\ \BBA {} Colliander%
}{%
Ahlgren%
\ \BBA {} Colliander%
}{%
{\protect \APACyear {2009}}%
}]{%
ahlgren2009}
\APACinsertmetastar {%
ahlgren2009}%
\begin{APACrefauthors}%
Ahlgren, P.%
\BCBT {}\ \BBA {} Colliander, C.%
\end{APACrefauthors}%
\unskip\
\newblock
\APACrefYearMonthDay{2009}{}{}.
\newblock
{\BBOQ}\APACrefatitle {Document–document similarity approaches and science
  mapping: Experimental comparison of five approaches} {Document–document
  similarity approaches and science mapping: Experimental comparison of five
  approaches}.{\BBCQ}
\newblock
\APACjournalVolNumPages{Journal of Informetrics}{3}{1}{49--63}.
\newblock
\begin{APACrefDOI} \doi{10.1016/j.joi.2008.11.003} \end{APACrefDOI}
\PrintBackRefs{\CurrentBib}

\bibitem [\protect \citeauthoryear {%
Alonso-Álvarez%
\ \BBA {} van Eck%
}{%
Alonso-Álvarez%
\ \BBA {} van Eck%
}{%
{\protect \APACyear {2025}}%
}]{%
alonso2025}
\APACinsertmetastar {%
alonso2025}%
\begin{APACrefauthors}%
Alonso-Álvarez, P.%
\BCBT {}\ \BBA {} van Eck, N\BPBI J.%
\end{APACrefauthors}%
\unskip\
\newblock
\APACrefYearMonthDay{2025}{}{}.
\newblock
{\BBOQ}\APACrefatitle {Coverage and metadata completeness and accuracy of
  {African} research publications in {OpenAlex}: A comparative analysis}
  {Coverage and metadata completeness and accuracy of {African} research
  publications in {OpenAlex}: A comparative analysis}.{\BBCQ}
\newblock
\APACjournalVolNumPages{Quantitative Science Studies}{6}{}{1336--1357}.
\newblock
\begin{APACrefDOI} \doi{10.1162/qss.a.396} \end{APACrefDOI}
\PrintBackRefs{\CurrentBib}

\bibitem [\protect \citeauthoryear {%
Aria%
\ \BBA {} Cuccurullo%
}{%
Aria%
\ \BBA {} Cuccurullo%
}{%
{\protect \APACyear {2017}}%
}]{%
aria2017}
\APACinsertmetastar {%
aria2017}%
\begin{APACrefauthors}%
Aria, M.%
\BCBT {}\ \BBA {} Cuccurullo, C.%
\end{APACrefauthors}%
\unskip\
\newblock
\APACrefYearMonthDay{2017}{}{}.
\newblock
{\BBOQ}\APACrefatitle {bibliometrix: An {R}-tool for comprehensive science
  mapping analysis} {bibliometrix: An {R}-tool for comprehensive science
  mapping analysis}.{\BBCQ}
\newblock
\APACjournalVolNumPages{Journal of Informetrics}{11}{4}{959--975}.
\newblock
\begin{APACrefDOI} \doi{10.1016/j.joi.2017.08.007} \end{APACrefDOI}
\PrintBackRefs{\CurrentBib}

\bibitem [\protect \citeauthoryear {%
Aria%
\ \BBA {} Cuccurullo%
}{%
Aria%
\ \BBA {} Cuccurullo%
}{%
{\protect \APACyear {2026}}%
}]{%
aria2026book}
\APACinsertmetastar {%
aria2026book}%
\begin{APACrefauthors}%
Aria, M.%
\BCBT {}\ \BBA {} Cuccurullo, C.%
\end{APACrefauthors}%
\unskip\
\newblock
\APACrefYear{2026}.
\newblock
\APACrefbtitle {Science Mapping Analysis: A Primer with {Biblioshiny}} {Science
  mapping analysis: A primer with {Biblioshiny}}.
\newblock
\APACaddressPublisher{}{McGraw-Hill}.
\PrintBackRefs{\CurrentBib}

\bibitem [\protect \citeauthoryear {%
Aria%
, Cuccurullo%
\BCBL {}\ \protect \BOthers {.}}{%
Aria%
, Cuccurullo%
\BCBL {}\ \protect \BOthers {.}}{%
{\protect \APACyear {2024}}%
}]{%
aria2024comp}
\APACinsertmetastar {%
aria2024comp}%
\begin{APACrefauthors}%
Aria, M.%
, Cuccurullo, C.%
, D’Aniello, L.%
, Misuraca, M.%
\BCBL {}\ \BBA {} Spano, M.%
\end{APACrefauthors}%
\unskip\
\newblock
\APACrefYearMonthDay{2024}{}{}.
\newblock
{\BBOQ}\APACrefatitle {Comparative science mapping: A novel conceptual
  structure analysis with metadata} {Comparative science mapping: A novel
  conceptual structure analysis with metadata}.{\BBCQ}
\newblock
\APACjournalVolNumPages{Scientometrics}{129}{11}{7055--7081}.
\newblock
\begin{APACrefDOI} \doi{10.1007/s11192-024-05161-6} \end{APACrefDOI}
\PrintBackRefs{\CurrentBib}

\bibitem [\protect \citeauthoryear {%
Aria%
, Cuccurullo%
\BCBL {}\ \BBA {} Spano%
}{%
Aria%
, Cuccurullo%
\BCBL {}\ \BBA {} Spano%
}{%
{\protect \APACyear {2026}}%
}]{%
aria2026data}
\APACinsertmetastar {%
aria2026data}%
\begin{APACrefauthors}%
Aria, M.%
, Cuccurullo, C.%
\BCBL {}\ \BBA {} Spano, M.%
\end{APACrefauthors}%
\unskip\
\newblock
\APACrefYearMonthDay{2026}{}{}.
\newblock
\APACrefbtitle {Replication material for ``{Multiplex} coupling: {Connecting}
  the intellectual base and the topics of a research field''
  {\upshape(\MakeUppercase{v}ersion 2)}} {Replication material for
  ``{Multiplex} coupling: {Connecting} the intellectual base and the topics of
  a research field'' {\upshape(\MakeUppercase{v}ersion 2)}}\ [Data set].
\newblock
\APAChowpublished {Zenodo}.
\newblock
\begin{APACrefDOI} \doi{10.5281/zenodo.23392228} \end{APACrefDOI}
\PrintBackRefs{\CurrentBib}

\bibitem [\protect \citeauthoryear {%
Aria%
, D’Aniello%
, Misuraca%
\BCBL {}\ \BBA {} Spano%
}{%
Aria%
, D’Aniello%
, Misuraca%
\BCBL {}\ \BBA {} Spano%
}{%
{\protect \APACyear {2026}}%
}]{%
aria2026b}
\APACinsertmetastar {%
aria2026b}%
\begin{APACrefauthors}%
Aria, M.%
, D’Aniello, L.%
, Misuraca, M.%
\BCBL {}\ \BBA {} Spano, M.%
\end{APACrefauthors}%
\unskip\
\newblock
\APACrefYearMonthDay{2026}{}{}.
\newblock
{\BBOQ}\APACrefatitle {Rethinking thematic evolution in science mapping: An
  integrated framework for longitudinal analysis} {Rethinking thematic
  evolution in science mapping: An integrated framework for longitudinal
  analysis}.{\BBCQ}
\newblock
\APACjournalVolNumPages{Journal of Informetrics}{20}{4}{Article 101877}.
\newblock
\begin{APACrefDOI} \doi{10.1016/j.joi.2026.101877} \end{APACrefDOI}
\PrintBackRefs{\CurrentBib}

\bibitem [\protect \citeauthoryear {%
Aria%
, D’Aniello%
\BCBL {}\ \BBA {} Spano%
}{%
Aria%
, D’Aniello%
\BCBL {}\ \BBA {} Spano%
}{%
{\protect \APACyear {2026}}%
}]{%
aria2026c}
\APACinsertmetastar {%
aria2026c}%
\begin{APACrefauthors}%
Aria, M.%
, D’Aniello, L.%
\BCBL {}\ \BBA {} Spano, M.%
\end{APACrefauthors}%
\unskip\
\newblock
\APACrefYearMonthDay{2026}{}{}.
\newblock
{\BBOQ}\APACrefatitle {A multi-phase reference matching algorithm for
  bibliometric analysis: Design, implementation, and evaluation} {A multi-phase
  reference matching algorithm for bibliometric analysis: Design,
  implementation, and evaluation}.{\BBCQ}
\newblock
\APACjournalVolNumPages{Scientometrics}{131}{9}{6693--6728}.
\newblock
\begin{APACrefDOI} \doi{10.1007/s11192-026-05763-2} \end{APACrefDOI}
\PrintBackRefs{\CurrentBib}

\bibitem [\protect \citeauthoryear {%
Aria%
, Le%
\BCBL {}\ \protect \BOthers {.}}{%
Aria%
, Le%
\BCBL {}\ \protect \BOthers {.}}{%
{\protect \APACyear {2024}}%
}]{%
aria2024oa}
\APACinsertmetastar {%
aria2024oa}%
\begin{APACrefauthors}%
Aria, M.%
, Le, T.%
, Cuccurullo, C.%
, Belfiore, A.%
\BCBL {}\ \BBA {} Choe, J.%
\end{APACrefauthors}%
\unskip\
\newblock
\APACrefYearMonthDay{2024}{}{}.
\newblock
{\BBOQ}\APACrefatitle {{openalexR}: An {R}-Tool for Collecting Bibliometric
  Data from {OpenAlex}} {{openalexR}: An {R}-tool for collecting bibliometric
  data from {OpenAlex}}.{\BBCQ}
\newblock
\APACjournalVolNumPages{The R Journal}{15}{4}{167--180}.
\newblock
\begin{APACrefDOI} \doi{10.32614/rj-2023-089} \end{APACrefDOI}
\PrintBackRefs{\CurrentBib}

\bibitem [\protect \citeauthoryear {%
Bascur%
\ \protect \BOthers {.}}{%
Bascur%
\ \protect \BOthers {.}}{%
{\protect \APACyear {2025}}%
}]{%
bascur2025}
\APACinsertmetastar {%
bascur2025}%
\begin{APACrefauthors}%
Bascur, J\BPBI P.%
, Verberne, S.%
, van Eck, N\BPBI J.%
\BCBL {}\ \BBA {} Waltman, L.%
\end{APACrefauthors}%
\unskip\
\newblock
\APACrefYearMonthDay{2025}{}{}.
\newblock
{\BBOQ}\APACrefatitle {Which topics are best represented by science maps? {An}
  analysis of clustering effectiveness for citation and text similarity
  networks} {Which topics are best represented by science maps? {An} analysis
  of clustering effectiveness for citation and text similarity
  networks}.{\BBCQ}
\newblock
\APACjournalVolNumPages{Scientometrics}{130}{2}{1181--1199}.
\newblock
\begin{APACrefDOI} \doi{10.1007/s11192-024-05218-6} \end{APACrefDOI}
\PrintBackRefs{\CurrentBib}

\bibitem [\protect \citeauthoryear {%
Battiston%
, Nicosia%
\BCBL {}\ \BBA {} Latora%
}{%
Battiston%
, Nicosia%
\BCBL {}\ \BBA {} Latora%
}{%
{\protect \APACyear {2014}}%
}]{%
battiston2014}
\APACinsertmetastar {%
battiston2014}%
\begin{APACrefauthors}%
Battiston, F.%
, Nicosia, V.%
\BCBL {}\ \BBA {} Latora, V.%
\end{APACrefauthors}%
\unskip\
\newblock
\APACrefYearMonthDay{2014}{}{}.
\newblock
{\BBOQ}\APACrefatitle {Structural measures for multiplex networks} {Structural
  measures for multiplex networks}.{\BBCQ}
\newblock
\APACjournalVolNumPages{Physical Review E}{89}{3}{Article 032804}.
\newblock
\begin{APACrefDOI} \doi{10.1103/physreve.89.032804} \end{APACrefDOI}
\PrintBackRefs{\CurrentBib}

\bibitem [\protect \citeauthoryear {%
Blondel%
\ \protect \BOthers {.}}{%
Blondel%
\ \protect \BOthers {.}}{%
{\protect \APACyear {2008}}%
}]{%
blondel2008}
\APACinsertmetastar {%
blondel2008}%
\begin{APACrefauthors}%
Blondel, V\BPBI D.%
, Guillaume, J\BHBI L.%
, Lambiotte, R.%
\BCBL {}\ \BBA {} Lefebvre, E.%
\end{APACrefauthors}%
\unskip\
\newblock
\APACrefYearMonthDay{2008}{}{}.
\newblock
{\BBOQ}\APACrefatitle {Fast unfolding of communities in large networks} {Fast
  unfolding of communities in large networks}.{\BBCQ}
\newblock
\APACjournalVolNumPages{Journal of Statistical Mechanics: Theory and
  Experiment}{2008}{10}{P10008}.
\newblock
\begin{APACrefDOI} \doi{10.1088/1742-5468/2008/10/p10008} \end{APACrefDOI}
\PrintBackRefs{\CurrentBib}

\bibitem [\protect \citeauthoryear {%
Boyack%
\ \BBA {} Klavans%
}{%
Boyack%
\ \BBA {} Klavans%
}{%
{\protect \APACyear {2010}}%
}]{%
boyack2010}
\APACinsertmetastar {%
boyack2010}%
\begin{APACrefauthors}%
Boyack, K\BPBI W.%
\BCBT {}\ \BBA {} Klavans, R.%
\end{APACrefauthors}%
\unskip\
\newblock
\APACrefYearMonthDay{2010}{}{}.
\newblock
{\BBOQ}\APACrefatitle {Co‐citation analysis, bibliographic coupling, and
  direct citation: Which citation approach represents the research front most
  accurately?} {Co‐citation analysis, bibliographic coupling, and direct
  citation: Which citation approach represents the research front most
  accurately?}{\BBCQ}
\newblock
\APACjournalVolNumPages{Journal of the American Society for Information Science
  and Technology}{61}{12}{2389--2404}.
\newblock
\begin{APACrefDOI} \doi{10.1002/asi.21419} \end{APACrefDOI}
\PrintBackRefs{\CurrentBib}

\bibitem [\protect \citeauthoryear {%
Boyack%
\ \BBA {} Klavans%
}{%
Boyack%
\ \BBA {} Klavans%
}{%
{\protect \APACyear {2020}}%
}]{%
boyack2020}
\APACinsertmetastar {%
boyack2020}%
\begin{APACrefauthors}%
Boyack, K\BPBI W.%
\BCBT {}\ \BBA {} Klavans, R.%
\end{APACrefauthors}%
\unskip\
\newblock
\APACrefYearMonthDay{2020}{}{}.
\newblock
{\BBOQ}\APACrefatitle {A comparison of large-scale science models based on
  textual, direct citation and hybrid relatedness} {A comparison of large-scale
  science models based on textual, direct citation and hybrid
  relatedness}.{\BBCQ}
\newblock
\APACjournalVolNumPages{Quantitative Science Studies}{1}{4}{1570--1585}.
\newblock
\begin{APACrefDOI} \doi{10.1162/qss_a_00085} \end{APACrefDOI}
\PrintBackRefs{\CurrentBib}

\bibitem [\protect \citeauthoryear {%
Braam%
, Moed%
\BCBL {}\ \BBA {} van Raan%
}{%
Braam%
, Moed%
\BCBL {}\ \BBA {} van Raan%
}{%
{\protect \APACyear {1991}}%
}]{%
braam1991}
\APACinsertmetastar {%
braam1991}%
\begin{APACrefauthors}%
Braam, R\BPBI R.%
, Moed, H\BPBI F.%
\BCBL {}\ \BBA {} van Raan, A\BPBI F\BPBI J.%
\end{APACrefauthors}%
\unskip\
\newblock
\APACrefYearMonthDay{1991}{}{}.
\newblock
{\BBOQ}\APACrefatitle {Mapping of science by combined co-citation and word
  analysis. {I}. {S}tructural aspects} {Mapping of science by combined
  co-citation and word analysis. {I}. {S}tructural aspects}.{\BBCQ}
\newblock
\APACjournalVolNumPages{Journal of the American Society for Information
  Science}{42}{4}{233--251}.
\newblock
\begin{APACrefDOI}
  \doi{10.1002/(sici)1097-4571(199105)42:4<233::aid-asi1>3.0.co;2-i}
  \end{APACrefDOI}
\PrintBackRefs{\CurrentBib}

\bibitem [\protect \citeauthoryear {%
Börner%
, Chen%
\BCBL {}\ \BBA {} Boyack%
}{%
Börner%
, Chen%
\BCBL {}\ \BBA {} Boyack%
}{%
{\protect \APACyear {2003}}%
}]{%
borner2003}
\APACinsertmetastar {%
borner2003}%
\begin{APACrefauthors}%
Börner, K.%
, Chen, C.%
\BCBL {}\ \BBA {} Boyack, K\BPBI W.%
\end{APACrefauthors}%
\unskip\
\newblock
\APACrefYearMonthDay{2003}{}{}.
\newblock
{\BBOQ}\APACrefatitle {Visualizing knowledge domains} {Visualizing knowledge
  domains}.{\BBCQ}
\newblock
\APACjournalVolNumPages{Annual Review of Information Science and
  Technology}{37}{1}{179--255}.
\newblock
\begin{APACrefDOI} \doi{10.1002/aris.1440370106} \end{APACrefDOI}
\PrintBackRefs{\CurrentBib}

\bibitem [\protect \citeauthoryear {%
Callon%
, Courtial%
\BCBL {}\ \BBA {} Laville%
}{%
Callon%
, Courtial%
\BCBL {}\ \BBA {} Laville%
}{%
{\protect \APACyear {1991}}%
}]{%
callon1991}
\APACinsertmetastar {%
callon1991}%
\begin{APACrefauthors}%
Callon, M.%
, Courtial, J\BHBI P.%
\BCBL {}\ \BBA {} Laville, F.%
\end{APACrefauthors}%
\unskip\
\newblock
\APACrefYearMonthDay{1991}{}{}.
\newblock
{\BBOQ}\APACrefatitle {Co-word analysis as a tool for describing the network of
  interactions between basic and technological research: The case of polymer
  chemsitry} {Co-word analysis as a tool for describing the network of
  interactions between basic and technological research: The case of polymer
  chemsitry}.{\BBCQ}
\newblock
\APACjournalVolNumPages{Scientometrics}{22}{1}{155--205}.
\newblock
\begin{APACrefDOI} \doi{10.1007/bf02019280} \end{APACrefDOI}
\PrintBackRefs{\CurrentBib}

\bibitem [\protect \citeauthoryear {%
Callon%
\ \protect \BOthers {.}}{%
Callon%
\ \protect \BOthers {.}}{%
{\protect \APACyear {1983}}%
}]{%
callon1983}
\APACinsertmetastar {%
callon1983}%
\begin{APACrefauthors}%
Callon, M.%
, Courtial, J\BHBI P.%
, Turner, W\BPBI A.%
\BCBL {}\ \BBA {} Bauin, S.%
\end{APACrefauthors}%
\unskip\
\newblock
\APACrefYearMonthDay{1983}{}{}.
\newblock
{\BBOQ}\APACrefatitle {From translations to problematic networks: An
  introduction to co-word analysis} {From translations to problematic networks:
  An introduction to co-word analysis}.{\BBCQ}
\newblock
\APACjournalVolNumPages{Social Science Information}{22}{2}{191--235}.
\newblock
\begin{APACrefDOI} \doi{10.1177/053901883022002003} \end{APACrefDOI}
\PrintBackRefs{\CurrentBib}

\bibitem [\protect \citeauthoryear {%
Cobo%
\ \protect \BOthers {.}}{%
Cobo%
\ \protect \BOthers {.}}{%
{\protect \APACyear {2011}}%
}]{%
cobo2011}
\APACinsertmetastar {%
cobo2011}%
\begin{APACrefauthors}%
Cobo, M\BPBI J.%
, López-Herrera, A\BPBI G.%
, Herrera-Viedma, E.%
\BCBL {}\ \BBA {} Herrera, F.%
\end{APACrefauthors}%
\unskip\
\newblock
\APACrefYearMonthDay{2011}{}{}.
\newblock
{\BBOQ}\APACrefatitle {Science mapping software tools: Review, analysis, and
  cooperative study among tools} {Science mapping software tools: Review,
  analysis, and cooperative study among tools}.{\BBCQ}
\newblock
\APACjournalVolNumPages{Journal of the American Society for Information Science
  and Technology}{62}{7}{1382--1402}.
\newblock
\begin{APACrefDOI} \doi{10.1002/asi.21525} \end{APACrefDOI}
\PrintBackRefs{\CurrentBib}

\bibitem [\protect \citeauthoryear {%
Culbert%
\ \protect \BOthers {.}}{%
Culbert%
\ \protect \BOthers {.}}{%
{\protect \APACyear {2025}}%
}]{%
culbert2025}
\APACinsertmetastar {%
culbert2025}%
\begin{APACrefauthors}%
Culbert, J\BPBI H.%
, Hobert, A.%
, Jahn, N.%
, Haupka, N.%
, Schmidt, M.%
, Donner, P.%
\BCBL {}\ \BBA {} Mayr, P.%
\end{APACrefauthors}%
\unskip\
\newblock
\APACrefYearMonthDay{2025}{}{}.
\newblock
{\BBOQ}\APACrefatitle {Reference coverage analysis of {OpenAlex} compared to
  {Web of Science} and {Scopus}} {Reference coverage analysis of {OpenAlex}
  compared to {Web of Science} and {Scopus}}.{\BBCQ}
\newblock
\APACjournalVolNumPages{Scientometrics}{130}{4}{2475--2492}.
\newblock
\begin{APACrefDOI} \doi{10.1007/s11192-025-05293-3} \end{APACrefDOI}
\PrintBackRefs{\CurrentBib}

\bibitem [\protect \citeauthoryear {%
Danon%
\ \protect \BOthers {.}}{%
Danon%
\ \protect \BOthers {.}}{%
{\protect \APACyear {2005}}%
}]{%
danon2005}
\APACinsertmetastar {%
danon2005}%
\begin{APACrefauthors}%
Danon, L.%
, Díaz-Guilera, A.%
, Duch, J.%
\BCBL {}\ \BBA {} Arenas, A.%
\end{APACrefauthors}%
\unskip\
\newblock
\APACrefYearMonthDay{2005}{}{}.
\newblock
{\BBOQ}\APACrefatitle {Comparing community structure identification} {Comparing
  community structure identification}.{\BBCQ}
\newblock
\APACjournalVolNumPages{Journal of Statistical Mechanics: Theory and
  Experiment}{2005}{09}{P09008}.
\newblock
\begin{APACrefDOI} \doi{10.1088/1742-5468/2005/09/p09008} \end{APACrefDOI}
\PrintBackRefs{\CurrentBib}

\bibitem [\protect \citeauthoryear {%
De~Domenico%
\ \protect \BOthers {.}}{%
De~Domenico%
\ \protect \BOthers {.}}{%
{\protect \APACyear {2015}}%
}]{%
dedomenico2015}
\APACinsertmetastar {%
dedomenico2015}%
\begin{APACrefauthors}%
De~Domenico, M.%
, Nicosia, V.%
, Arenas, A.%
\BCBL {}\ \BBA {} Latora, V.%
\end{APACrefauthors}%
\unskip\
\newblock
\APACrefYearMonthDay{2015}{}{}.
\newblock
{\BBOQ}\APACrefatitle {Structural reducibility of multilayer networks}
  {Structural reducibility of multilayer networks}.{\BBCQ}
\newblock
\APACjournalVolNumPages{Nature Communications}{6}{1}{Article 6864}.
\newblock
\begin{APACrefDOI} \doi{10.1038/ncomms7864} \end{APACrefDOI}
\PrintBackRefs{\CurrentBib}

\bibitem [\protect \citeauthoryear {%
Geison%
}{%
Geison%
}{%
{\protect \APACyear {1981}}%
}]{%
geison1981}
\APACinsertmetastar {%
geison1981}%
\begin{APACrefauthors}%
Geison, G\BPBI L.%
\end{APACrefauthors}%
\unskip\
\newblock
\APACrefYearMonthDay{1981}{}{}.
\newblock
{\BBOQ}\APACrefatitle {Scientific Change, Emerging Specialties, and Research
  Schools} {Scientific change, emerging specialties, and research
  schools}.{\BBCQ}
\newblock
\APACjournalVolNumPages{History of Science}{19}{1}{20--40}.
\newblock
\begin{APACrefDOI} \doi{10.1177/007327538101900103} \end{APACrefDOI}
\PrintBackRefs{\CurrentBib}

\bibitem [\protect \citeauthoryear {%
Glänzel%
\ \BBA {} Thijs%
}{%
Glänzel%
\ \BBA {} Thijs%
}{%
{\protect \APACyear {2012}}%
}]{%
glanzel2012}
\APACinsertmetastar {%
glanzel2012}%
\begin{APACrefauthors}%
Glänzel, W.%
\BCBT {}\ \BBA {} Thijs, B.%
\end{APACrefauthors}%
\unskip\
\newblock
\APACrefYearMonthDay{2012}{}{}.
\newblock
{\BBOQ}\APACrefatitle {Using ‘core documents’ for detecting and labelling
  new emerging topics} {Using ‘core documents’ for detecting and labelling
  new emerging topics}.{\BBCQ}
\newblock
\APACjournalVolNumPages{Scientometrics}{91}{2}{399--416}.
\newblock
\begin{APACrefDOI} \doi{10.1007/s11192-011-0591-7} \end{APACrefDOI}
\PrintBackRefs{\CurrentBib}

\bibitem [\protect \citeauthoryear {%
Haberman%
}{%
Haberman%
}{%
{\protect \APACyear {1973}}%
}]{%
haberman1973}
\APACinsertmetastar {%
haberman1973}%
\begin{APACrefauthors}%
Haberman, S\BPBI J.%
\end{APACrefauthors}%
\unskip\
\newblock
\APACrefYearMonthDay{1973}{}{}.
\newblock
{\BBOQ}\APACrefatitle {The Analysis of Residuals in Cross-Classified Tables}
  {The analysis of residuals in cross-classified tables}.{\BBCQ}
\newblock
\APACjournalVolNumPages{Biometrics}{29}{1}{205--220}.
\newblock
\begin{APACrefDOI} \doi{10.2307/2529686} \end{APACrefDOI}
\PrintBackRefs{\CurrentBib}

\bibitem [\protect \citeauthoryear {%
Havemann%
}{%
Havemann%
}{%
{\protect \APACyear {2021}}%
}]{%
havemann2021}
\APACinsertmetastar {%
havemann2021}%
\begin{APACrefauthors}%
Havemann, F.%
\end{APACrefauthors}%
\unskip\
\newblock
\APACrefYearMonthDay{2021}{}{}.
\newblock
{\BBOQ}\APACrefatitle {Topics as clusters of citation links to highly cited
  sources: The case of research on international relations} {Topics as clusters
  of citation links to highly cited sources: The case of research on
  international relations}.{\BBCQ}
\newblock
\APACjournalVolNumPages{Quantitative Science Studies}{2}{1}{204--223}.
\newblock
\begin{APACrefDOI} \doi{10.1162/qss_a_00108} \end{APACrefDOI}
\PrintBackRefs{\CurrentBib}

\bibitem [\protect \citeauthoryear {%
Held%
}{%
Held%
}{%
{\protect \APACyear {2022}}%
}]{%
held2022}
\APACinsertmetastar {%
held2022}%
\begin{APACrefauthors}%
Held, M.%
\end{APACrefauthors}%
\unskip\
\newblock
\APACrefYearMonthDay{2022}{}{}.
\newblock
{\BBOQ}\APACrefatitle {Know thy tools! {Limits} of popular algorithms used for
  topic reconstruction} {Know thy tools! {Limits} of popular algorithms used
  for topic reconstruction}.{\BBCQ}
\newblock
\APACjournalVolNumPages{Quantitative Science Studies}{3}{4}{1054--1078}.
\newblock
\begin{APACrefDOI} \doi{10.1162/qss_a_00217} \end{APACrefDOI}
\PrintBackRefs{\CurrentBib}

\bibitem [\protect \citeauthoryear {%
Held%
\ \BBA {} Velden%
}{%
Held%
\ \BBA {} Velden%
}{%
{\protect \APACyear {2022}}%
}]{%
held2022velden}
\APACinsertmetastar {%
held2022velden}%
\begin{APACrefauthors}%
Held, M.%
\BCBT {}\ \BBA {} Velden, T.%
\end{APACrefauthors}%
\unskip\
\newblock
\APACrefYearMonthDay{2022}{}{}.
\newblock
{\BBOQ}\APACrefatitle {How to interpret algorithmically constructed topical
  structures of scientific fields? {A} case study of citation-based mappings of
  the research specialty of invasion biology} {How to interpret algorithmically
  constructed topical structures of scientific fields? {A} case study of
  citation-based mappings of the research specialty of invasion
  biology}.{\BBCQ}
\newblock
\APACjournalVolNumPages{Quantitative Science Studies}{3}{3}{651--671}.
\newblock
\begin{APACrefDOI} \doi{10.1162/qss_a_00194} \end{APACrefDOI}
\PrintBackRefs{\CurrentBib}

\bibitem [\protect \citeauthoryear {%
Hubert%
\ \BBA {} Arabie%
}{%
Hubert%
\ \BBA {} Arabie%
}{%
{\protect \APACyear {1985}}%
}]{%
hubert1985}
\APACinsertmetastar {%
hubert1985}%
\begin{APACrefauthors}%
Hubert, L.%
\BCBT {}\ \BBA {} Arabie, P.%
\end{APACrefauthors}%
\unskip\
\newblock
\APACrefYearMonthDay{1985}{}{}.
\newblock
{\BBOQ}\APACrefatitle {Comparing partitions} {Comparing partitions}.{\BBCQ}
\newblock
\APACjournalVolNumPages{Journal of Classification}{2}{1}{193--218}.
\newblock
\begin{APACrefDOI} \doi{10.1007/bf01908075} \end{APACrefDOI}
\PrintBackRefs{\CurrentBib}

\bibitem [\protect \citeauthoryear {%
Hubert%
\ \BBA {} Schultz%
}{%
Hubert%
\ \BBA {} Schultz%
}{%
{\protect \APACyear {1976}}%
}]{%
hubert1976}
\APACinsertmetastar {%
hubert1976}%
\begin{APACrefauthors}%
Hubert, L.%
\BCBT {}\ \BBA {} Schultz, J.%
\end{APACrefauthors}%
\unskip\
\newblock
\APACrefYearMonthDay{1976}{}{}.
\newblock
{\BBOQ}\APACrefatitle {QUADRATIC ASSIGNMENT AS A GENERAL DATA ANALYSIS
  STRATEGY} {Quadratic assignment as a general data analysis strategy}.{\BBCQ}
\newblock
\APACjournalVolNumPages{British Journal of Mathematical and Statistical
  Psychology}{29}{2}{190--241}.
\newblock
\begin{APACrefDOI} \doi{10.1111/j.2044-8317.1976.tb00714.x} \end{APACrefDOI}
\PrintBackRefs{\CurrentBib}

\bibitem [\protect \citeauthoryear {%
Janssens%
, Glänzel%
\BCBL {}\ \BBA {} De~Moor%
}{%
Janssens%
, Glänzel%
\BCBL {}\ \BBA {} De~Moor%
}{%
{\protect \APACyear {2008}}%
}]{%
janssens2008}
\APACinsertmetastar {%
janssens2008}%
\begin{APACrefauthors}%
Janssens, F.%
, Glänzel, W.%
\BCBL {}\ \BBA {} De~Moor, B.%
\end{APACrefauthors}%
\unskip\
\newblock
\APACrefYearMonthDay{2008}{}{}.
\newblock
{\BBOQ}\APACrefatitle {A hybrid mapping of information science} {A hybrid
  mapping of information science}.{\BBCQ}
\newblock
\APACjournalVolNumPages{Scientometrics}{75}{3}{607--631}.
\newblock
\begin{APACrefDOI} \doi{10.1007/s11192-007-2002-7} \end{APACrefDOI}
\PrintBackRefs{\CurrentBib}

\bibitem [\protect \citeauthoryear {%
Kessler%
}{%
Kessler%
}{%
{\protect \APACyear {1963}}%
}]{%
kessler1963}
\APACinsertmetastar {%
kessler1963}%
\begin{APACrefauthors}%
Kessler, M\BPBI M.%
\end{APACrefauthors}%
\unskip\
\newblock
\APACrefYearMonthDay{1963}{}{}.
\newblock
{\BBOQ}\APACrefatitle {Bibliographic coupling between scientific papers}
  {Bibliographic coupling between scientific papers}.{\BBCQ}
\newblock
\APACjournalVolNumPages{American Documentation}{14}{1}{10--25}.
\newblock
\begin{APACrefDOI} \doi{10.1002/asi.5090140103} \end{APACrefDOI}
\PrintBackRefs{\CurrentBib}

\bibitem [\protect \citeauthoryear {%
Kivelä%
\ \protect \BOthers {.}}{%
Kivelä%
\ \protect \BOthers {.}}{%
{\protect \APACyear {2014}}%
}]{%
kivela2014}
\APACinsertmetastar {%
kivela2014}%
\begin{APACrefauthors}%
Kivelä, M.%
, Arenas, A.%
, Barthélemy, M.%
, Gleeson, J\BPBI P.%
, Moreno, Y.%
\BCBL {}\ \BBA {} Porter, M\BPBI A.%
\end{APACrefauthors}%
\unskip\
\newblock
\APACrefYearMonthDay{2014}{}{}.
\newblock
{\BBOQ}\APACrefatitle {Multilayer networks} {Multilayer networks}.{\BBCQ}
\newblock
\APACjournalVolNumPages{Journal of Complex Networks}{2}{3}{203--271}.
\newblock
\begin{APACrefDOI} \doi{10.1093/comnet/cnu016} \end{APACrefDOI}
\PrintBackRefs{\CurrentBib}

\bibitem [\protect \citeauthoryear {%
Klavans%
\ \BBA {} Boyack%
}{%
Klavans%
\ \BBA {} Boyack%
}{%
{\protect \APACyear {2017}}%
}]{%
klavans2017}
\APACinsertmetastar {%
klavans2017}%
\begin{APACrefauthors}%
Klavans, R.%
\BCBT {}\ \BBA {} Boyack, K\BPBI W.%
\end{APACrefauthors}%
\unskip\
\newblock
\APACrefYearMonthDay{2017}{}{}.
\newblock
{\BBOQ}\APACrefatitle {Which Type of Citation Analysis Generates the Most
  Accurate Taxonomy of Scientific and Technical Knowledge?} {Which type of
  citation analysis generates the most accurate taxonomy of scientific and
  technical knowledge?}{\BBCQ}
\newblock
\APACjournalVolNumPages{Journal of the Association for Information Science and
  Technology}{68}{4}{984--998}.
\newblock
\begin{APACrefDOI} \doi{10.1002/asi.23734} \end{APACrefDOI}
\PrintBackRefs{\CurrentBib}

\bibitem [\protect \citeauthoryear {%
Krackhardt%
}{%
Krackhardt%
}{%
{\protect \APACyear {1988}}%
}]{%
krackhardt1988}
\APACinsertmetastar {%
krackhardt1988}%
\begin{APACrefauthors}%
Krackhardt, D.%
\end{APACrefauthors}%
\unskip\
\newblock
\APACrefYearMonthDay{1988}{}{}.
\newblock
{\BBOQ}\APACrefatitle {Predicting with networks: Nonparametric multiple
  regression analysis of dyadic data} {Predicting with networks: Nonparametric
  multiple regression analysis of dyadic data}.{\BBCQ}
\newblock
\APACjournalVolNumPages{Social Networks}{10}{4}{359--381}.
\newblock
\begin{APACrefDOI} \doi{10.1016/0378-8733(88)90004-4} \end{APACrefDOI}
\PrintBackRefs{\CurrentBib}

\bibitem [\protect \citeauthoryear {%
Kuhn%
}{%
Kuhn%
}{%
{\protect \APACyear {1970}}%
}]{%
kuhn1970}
\APACinsertmetastar {%
kuhn1970}%
\begin{APACrefauthors}%
Kuhn, T\BPBI S.%
\end{APACrefauthors}%
\unskip\
\newblock
\APACrefYear{1970}.
\newblock
\APACrefbtitle {The Structure of Scientific Revolutions} {The structure of
  scientific revolutions}\ (\PrintOrdinal{2nd}\ \BEd).
\newblock
\APACaddressPublisher{}{University of Chicago Press}.
\PrintBackRefs{\CurrentBib}

\bibitem [\protect \citeauthoryear {%
Lancichinetti%
\ \BBA {} Fortunato%
}{%
Lancichinetti%
\ \BBA {} Fortunato%
}{%
{\protect \APACyear {2012}}%
}]{%
lancichinetti2012}
\APACinsertmetastar {%
lancichinetti2012}%
\begin{APACrefauthors}%
Lancichinetti, A.%
\BCBT {}\ \BBA {} Fortunato, S.%
\end{APACrefauthors}%
\unskip\
\newblock
\APACrefYearMonthDay{2012}{}{}.
\newblock
{\BBOQ}\APACrefatitle {Consensus clustering in complex networks} {Consensus
  clustering in complex networks}.{\BBCQ}
\newblock
\APACjournalVolNumPages{Scientific Reports}{2}{1}{Article 336}.
\newblock
\begin{APACrefDOI} \doi{10.1038/srep00336} \end{APACrefDOI}
\PrintBackRefs{\CurrentBib}

\bibitem [\protect \citeauthoryear {%
Liu%
\ \protect \BOthers {.}}{%
Liu%
\ \protect \BOthers {.}}{%
{\protect \APACyear {2010}}%
}]{%
liu2010}
\APACinsertmetastar {%
liu2010}%
\begin{APACrefauthors}%
Liu, X.%
, Yu, S.%
, Janssens, F.%
, Glänzel, W.%
, Moreau, Y.%
\BCBL {}\ \BBA {} De~Moor, B.%
\end{APACrefauthors}%
\unskip\
\newblock
\APACrefYearMonthDay{2010}{}{}.
\newblock
{\BBOQ}\APACrefatitle {Weighted hybrid clustering by combining text mining and
  bibliometrics on a large‐scale journal database} {Weighted hybrid
  clustering by combining text mining and bibliometrics on a large‐scale
  journal database}.{\BBCQ}
\newblock
\APACjournalVolNumPages{Journal of the American Society for Information Science
  and Technology}{61}{6}{1105--1119}.
\newblock
\begin{APACrefDOI} \doi{10.1002/asi.21312} \end{APACrefDOI}
\PrintBackRefs{\CurrentBib}

\bibitem [\protect \citeauthoryear {%
Morrell%
}{%
Morrell%
}{%
{\protect \APACyear {1972}}%
}]{%
morrell1972}
\APACinsertmetastar {%
morrell1972}%
\begin{APACrefauthors}%
Morrell, J\BPBI B.%
\end{APACrefauthors}%
\unskip\
\newblock
\APACrefYearMonthDay{1972}{}{}.
\newblock
{\BBOQ}\APACrefatitle {The Chemist Breeders: The Research Schools of {Liebig}
  and {Thomas Thomson}} {The chemist breeders: The research schools of {Liebig}
  and {Thomas Thomson}}.{\BBCQ}
\newblock
\APACjournalVolNumPages{Ambix}{19}{1}{1--46}.
\newblock
\begin{APACrefDOI} \doi{10.1179/amb.1972.19.1.1} \end{APACrefDOI}
\PrintBackRefs{\CurrentBib}

\bibitem [\protect \citeauthoryear {%
Mucha%
\ \protect \BOthers {.}}{%
Mucha%
\ \protect \BOthers {.}}{%
{\protect \APACyear {2010}}%
}]{%
mucha2010}
\APACinsertmetastar {%
mucha2010}%
\begin{APACrefauthors}%
Mucha, P\BPBI J.%
, Richardson, T.%
, Macon, K.%
, Porter, M\BPBI A.%
\BCBL {}\ \BBA {} Onnela, J\BHBI P.%
\end{APACrefauthors}%
\unskip\
\newblock
\APACrefYearMonthDay{2010}{}{}.
\newblock
{\BBOQ}\APACrefatitle {Community Structure in Time-Dependent, Multiscale, and
  Multiplex Networks} {Community structure in time-dependent, multiscale, and
  multiplex networks}.{\BBCQ}
\newblock
\APACjournalVolNumPages{Science}{328}{5980}{876--878}.
\newblock
\begin{APACrefDOI} \doi{10.1126/science.1184819} \end{APACrefDOI}
\PrintBackRefs{\CurrentBib}

\bibitem [\protect \citeauthoryear {%
Persson%
}{%
Persson%
}{%
{\protect \APACyear {1994}}%
}]{%
persson1994}
\APACinsertmetastar {%
persson1994}%
\begin{APACrefauthors}%
Persson, O.%
\end{APACrefauthors}%
\unskip\
\newblock
\APACrefYearMonthDay{1994}{}{}.
\newblock
{\BBOQ}\APACrefatitle {The intellectual base and research fronts of {JASIS}
  1986--1990} {The intellectual base and research fronts of {JASIS}
  1986--1990}.{\BBCQ}
\newblock
\APACjournalVolNumPages{Journal of the American Society for Information
  Science}{45}{1}{31--38}.
\newblock
\begin{APACrefDOI}
  \doi{10.1002/(sici)1097-4571(199401)45:1<31::aid-asi4>3.0.co;2-g}
  \end{APACrefDOI}
\PrintBackRefs{\CurrentBib}

\bibitem [\protect \citeauthoryear {%
Priem%
, Piwowar%
\BCBL {}\ \BBA {} Orr%
}{%
Priem%
, Piwowar%
\BCBL {}\ \BBA {} Orr%
}{%
{\protect \APACyear {2022}}%
}]{%
priem2022}
\APACinsertmetastar {%
priem2022}%
\begin{APACrefauthors}%
Priem, J.%
, Piwowar, H.%
\BCBL {}\ \BBA {} Orr, R.%
\end{APACrefauthors}%
\unskip\
\newblock
\APACrefYearMonthDay{2022}{}{}.
\newblock
\APACrefbtitle {{OpenAlex}: A fully-open index of scholarly works, authors,
  venues, institutions, and concepts.} {{OpenAlex}: A fully-open index of
  scholarly works, authors, venues, institutions, and concepts.}
\newblock
\APACaddressPublisher{}{arXiv}.
\newblock
\begin{APACrefDOI} \doi{10.48550/arXiv.2205.01833} \end{APACrefDOI}
\PrintBackRefs{\CurrentBib}

\bibitem [\protect \citeauthoryear {%
Salton%
\ \BBA {} McGill%
}{%
Salton%
\ \BBA {} McGill%
}{%
{\protect \APACyear {1983}}%
}]{%
salton1983}
\APACinsertmetastar {%
salton1983}%
\begin{APACrefauthors}%
Salton, G.%
\BCBT {}\ \BBA {} McGill, M\BPBI J.%
\end{APACrefauthors}%
\unskip\
\newblock
\APACrefYear{1983}.
\newblock
\APACrefbtitle {Introduction to Modern Information Retrieval} {Introduction to
  modern information retrieval}.
\newblock
\APACaddressPublisher{}{McGraw-Hill}.
\PrintBackRefs{\CurrentBib}

\bibitem [\protect \citeauthoryear {%
Sjögårde%
\ \BBA {} Ahlgren%
}{%
Sjögårde%
\ \BBA {} Ahlgren%
}{%
{\protect \APACyear {2020}}%
}]{%
sjogarde2020}
\APACinsertmetastar {%
sjogarde2020}%
\begin{APACrefauthors}%
Sjögårde, P.%
\BCBT {}\ \BBA {} Ahlgren, P.%
\end{APACrefauthors}%
\unskip\
\newblock
\APACrefYearMonthDay{2020}{}{}.
\newblock
{\BBOQ}\APACrefatitle {Granularity of algorithmically constructed
  publication-level classifications of research publications: Identification of
  specialties} {Granularity of algorithmically constructed publication-level
  classifications of research publications: Identification of
  specialties}.{\BBCQ}
\newblock
\APACjournalVolNumPages{Quantitative Science Studies}{1}{1}{207--238}.
\newblock
\begin{APACrefDOI} \doi{10.1162/qss_a_00004} \end{APACrefDOI}
\PrintBackRefs{\CurrentBib}

\bibitem [\protect \citeauthoryear {%
Small%
}{%
Small%
}{%
{\protect \APACyear {1973}}%
}]{%
small1973}
\APACinsertmetastar {%
small1973}%
\begin{APACrefauthors}%
Small, H.%
\end{APACrefauthors}%
\unskip\
\newblock
\APACrefYearMonthDay{1973}{}{}.
\newblock
{\BBOQ}\APACrefatitle {Co‐citation in the scientific literature: A new
  measure of the relationship between two documents} {Co‐citation in the
  scientific literature: A new measure of the relationship between two
  documents}.{\BBCQ}
\newblock
\APACjournalVolNumPages{Journal of the American Society for Information
  Science}{24}{4}{265--269}.
\newblock
\begin{APACrefDOI} \doi{10.1002/asi.4630240406} \end{APACrefDOI}
\PrintBackRefs{\CurrentBib}

\bibitem [\protect \citeauthoryear {%
Strona%
\ \protect \BOthers {.}}{%
Strona%
\ \protect \BOthers {.}}{%
{\protect \APACyear {2014}}%
}]{%
strona2014}
\APACinsertmetastar {%
strona2014}%
\begin{APACrefauthors}%
Strona, G.%
, Nappo, D.%
, Boccacci, F.%
, Fattorini, S.%
\BCBL {}\ \BBA {} San-Miguel-Ayanz, J.%
\end{APACrefauthors}%
\unskip\
\newblock
\APACrefYearMonthDay{2014}{}{}.
\newblock
{\BBOQ}\APACrefatitle {A fast and unbiased procedure to randomize ecological
  binary matrices with fixed row and column totals} {A fast and unbiased
  procedure to randomize ecological binary matrices with fixed row and column
  totals}.{\BBCQ}
\newblock
\APACjournalVolNumPages{Nature Communications}{5}{1}{Article 4114}.
\newblock
\begin{APACrefDOI} \doi{10.1038/ncomms5114} \end{APACrefDOI}
\PrintBackRefs{\CurrentBib}

\bibitem [\protect \citeauthoryear {%
Thijs%
\ \BBA {} Glänzel%
}{%
Thijs%
\ \BBA {} Glänzel%
}{%
{\protect \APACyear {2018}}%
}]{%
thijs2018}
\APACinsertmetastar {%
thijs2018}%
\begin{APACrefauthors}%
Thijs, B.%
\BCBT {}\ \BBA {} Glänzel, W.%
\end{APACrefauthors}%
\unskip\
\newblock
\APACrefYearMonthDay{2018}{}{}.
\newblock
{\BBOQ}\APACrefatitle {The contribution of the lexical component in hybrid
  clustering, the case of four decades of “{Scientometrics}”} {The
  contribution of the lexical component in hybrid clustering, the case of four
  decades of “{Scientometrics}”}.{\BBCQ}
\newblock
\APACjournalVolNumPages{Scientometrics}{115}{1}{21--33}.
\newblock
\begin{APACrefDOI} \doi{10.1007/s11192-018-2659-0} \end{APACrefDOI}
\PrintBackRefs{\CurrentBib}

\bibitem [\protect \citeauthoryear {%
Thijs%
, Schiebel%
\BCBL {}\ \BBA {} Glänzel%
}{%
Thijs%
, Schiebel%
\BCBL {}\ \BBA {} Glänzel%
}{%
{\protect \APACyear {2013}}%
}]{%
thijs2013}
\APACinsertmetastar {%
thijs2013}%
\begin{APACrefauthors}%
Thijs, B.%
, Schiebel, E.%
\BCBL {}\ \BBA {} Glänzel, W.%
\end{APACrefauthors}%
\unskip\
\newblock
\APACrefYearMonthDay{2013}{}{}.
\newblock
{\BBOQ}\APACrefatitle {Do second-order similarities provide added-value in a
  hybrid approach?} {Do second-order similarities provide added-value in a
  hybrid approach?}{\BBCQ}
\newblock
\APACjournalVolNumPages{Scientometrics}{96}{3}{667--677}.
\newblock
\begin{APACrefDOI} \doi{10.1007/s11192-012-0896-1} \end{APACrefDOI}
\PrintBackRefs{\CurrentBib}

\bibitem [\protect \citeauthoryear {%
Traag%
, Waltman%
\BCBL {}\ \BBA {} van Eck%
}{%
Traag%
, Waltman%
\BCBL {}\ \BBA {} van Eck%
}{%
{\protect \APACyear {2019}}%
}]{%
traag2019}
\APACinsertmetastar {%
traag2019}%
\begin{APACrefauthors}%
Traag, V\BPBI A.%
, Waltman, L.%
\BCBL {}\ \BBA {} van Eck, N\BPBI J.%
\end{APACrefauthors}%
\unskip\
\newblock
\APACrefYearMonthDay{2019}{}{}.
\newblock
{\BBOQ}\APACrefatitle {From {Louvain} to {Leiden}: Guaranteeing well-connected
  communities} {From {Louvain} to {Leiden}: Guaranteeing well-connected
  communities}.{\BBCQ}
\newblock
\APACjournalVolNumPages{Scientific Reports}{9}{1}{Article 5233}.
\newblock
\begin{APACrefDOI} \doi{10.1038/s41598-019-41695-z} \end{APACrefDOI}
\PrintBackRefs{\CurrentBib}

\bibitem [\protect \citeauthoryear {%
van Eck%
\ \BBA {} Waltman%
}{%
van Eck%
\ \BBA {} Waltman%
}{%
{\protect \APACyear {2009}}%
}]{%
vaneck2009}
\APACinsertmetastar {%
vaneck2009}%
\begin{APACrefauthors}%
van Eck, N\BPBI J.%
\BCBT {}\ \BBA {} Waltman, L.%
\end{APACrefauthors}%
\unskip\
\newblock
\APACrefYearMonthDay{2009}{}{}.
\newblock
{\BBOQ}\APACrefatitle {How to normalize cooccurrence data? {An} analysis of
  some well‐known similarity measures} {How to normalize cooccurrence data?
  {An} analysis of some well‐known similarity measures}.{\BBCQ}
\newblock
\APACjournalVolNumPages{Journal of the American Society for Information Science
  and Technology}{60}{8}{1635--1651}.
\newblock
\begin{APACrefDOI} \doi{10.1002/asi.21075} \end{APACrefDOI}
\PrintBackRefs{\CurrentBib}

\bibitem [\protect \citeauthoryear {%
Waltman%
\ \protect \BOthers {.}}{%
Waltman%
\ \protect \BOthers {.}}{%
{\protect \APACyear {2020}}%
}]{%
waltman2020}
\APACinsertmetastar {%
waltman2020}%
\begin{APACrefauthors}%
Waltman, L.%
, Boyack, K\BPBI W.%
, Colavizza, G.%
\BCBL {}\ \BBA {} van Eck, N\BPBI J.%
\end{APACrefauthors}%
\unskip\
\newblock
\APACrefYearMonthDay{2020}{}{}.
\newblock
{\BBOQ}\APACrefatitle {A principled methodology for comparing relatedness
  measures for clustering publications} {A principled methodology for comparing
  relatedness measures for clustering publications}.{\BBCQ}
\newblock
\APACjournalVolNumPages{Quantitative Science Studies}{1}{2}{691--713}.
\newblock
\begin{APACrefDOI} \doi{10.1162/qss_a_00035} \end{APACrefDOI}
\PrintBackRefs{\CurrentBib}

\bibitem [\protect \citeauthoryear {%
Xie%
\ \BBA {} Waltman%
}{%
Xie%
\ \BBA {} Waltman%
}{%
{\protect \APACyear {2025}}%
}]{%
xie2025}
\APACinsertmetastar {%
xie2025}%
\begin{APACrefauthors}%
Xie, Q.%
\BCBT {}\ \BBA {} Waltman, L.%
\end{APACrefauthors}%
\unskip\
\newblock
\APACrefYearMonthDay{2025}{}{}.
\newblock
{\BBOQ}\APACrefatitle {A comparison of citation-based clustering and topic
  modeling for science mapping} {A comparison of citation-based clustering and
  topic modeling for science mapping}.{\BBCQ}
\newblock
\APACjournalVolNumPages{Scientometrics}{130}{5}{2497--2522}.
\newblock
\begin{APACrefDOI} \doi{10.1007/s11192-025-05324-z} \end{APACrefDOI}
\PrintBackRefs{\CurrentBib}

\bibitem [\protect \citeauthoryear {%
Yan%
\ \BBA {} Ding%
}{%
Yan%
\ \BBA {} Ding%
}{%
{\protect \APACyear {2012}}%
}]{%
yan2012}
\APACinsertmetastar {%
yan2012}%
\begin{APACrefauthors}%
Yan, E.%
\BCBT {}\ \BBA {} Ding, Y.%
\end{APACrefauthors}%
\unskip\
\newblock
\APACrefYearMonthDay{2012}{}{}.
\newblock
{\BBOQ}\APACrefatitle {Scholarly network similarities: How bibliographic
  coupling networks, citation networks, cocitation networks, topical networks,
  coauthorship networks, and coword networks relate to each other} {Scholarly
  network similarities: How bibliographic coupling networks, citation networks,
  cocitation networks, topical networks, coauthorship networks, and coword
  networks relate to each other}.{\BBCQ}
\newblock
\APACjournalVolNumPages{Journal of the American Society for Information Science
  and Technology}{63}{7}{1313--1326}.
\newblock
\begin{APACrefDOI} \doi{10.1002/asi.22680} \end{APACrefDOI}
\PrintBackRefs{\CurrentBib}

\bibitem [\protect \citeauthoryear {%
Yang%
\ \protect \BOthers {.}}{%
Yang%
\ \protect \BOthers {.}}{%
{\protect \APACyear {2016}}%
}]{%
yang2016}
\APACinsertmetastar {%
yang2016}%
\begin{APACrefauthors}%
Yang, S.%
, Han, R.%
, Wolfram, D.%
\BCBL {}\ \BBA {} Zhao, Y.%
\end{APACrefauthors}%
\unskip\
\newblock
\APACrefYearMonthDay{2016}{}{}.
\newblock
{\BBOQ}\APACrefatitle {Visualizing the intellectual structure of information
  science (2006–2015): Introducing author keyword coupling analysis}
  {Visualizing the intellectual structure of information science (2006–2015):
  Introducing author keyword coupling analysis}.{\BBCQ}
\newblock
\APACjournalVolNumPages{Journal of Informetrics}{10}{1}{132--150}.
\newblock
\begin{APACrefDOI} \doi{10.1016/j.joi.2015.12.003} \end{APACrefDOI}
\PrintBackRefs{\CurrentBib}

\bibitem [\protect \citeauthoryear {%
Zhou%
, Liang%
\BCBL {}\ \BBA {} Acuna%
}{%
Zhou%
, Liang%
\BCBL {}\ \BBA {} Acuna%
}{%
{\protect \APACyear {2026}}%
}]{%
zhou2026}
\APACinsertmetastar {%
zhou2026}%
\begin{APACrefauthors}%
Zhou, H.%
, Liang, L.%
\BCBL {}\ \BBA {} Acuna, D\BPBI E.%
\end{APACrefauthors}%
\unskip\
\newblock
\APACrefYearMonthDay{2026}{}{}.
\newblock
{\BBOQ}\APACrefatitle {Widespread reference missingness disparities in open
  scholarly metadata} {Widespread reference missingness disparities in open
  scholarly metadata}.{\BBCQ}
\newblock
\APACjournalVolNumPages{Quantitative Science Studies}{7}{}{53--65}.
\newblock
\begin{APACrefDOI} \doi{10.1162/qss.a.400} \end{APACrefDOI}
\PrintBackRefs{\CurrentBib}

\bibitem [\protect \citeauthoryear {%
Zupic%
\ \BBA {} Čater%
}{%
Zupic%
\ \BBA {} Čater%
}{%
{\protect \APACyear {2015}}%
}]{%
zupic2015}
\APACinsertmetastar {%
zupic2015}%
\begin{APACrefauthors}%
Zupic, I.%
\BCBT {}\ \BBA {} Čater, T.%
\end{APACrefauthors}%
\unskip\
\newblock
\APACrefYearMonthDay{2015}{}{}.
\newblock
{\BBOQ}\APACrefatitle {Bibliometric Methods in Management and Organization}
  {Bibliometric methods in management and organization}.{\BBCQ}
\newblock
\APACjournalVolNumPages{Organizational Research Methods}{18}{3}{429--472}.
\newblock
\begin{APACrefDOI} \doi{10.1177/1094428114562629} \end{APACrefDOI}
\PrintBackRefs{\CurrentBib}

\end{thebibliography}

\clearpage
\includepdf[pages=-]{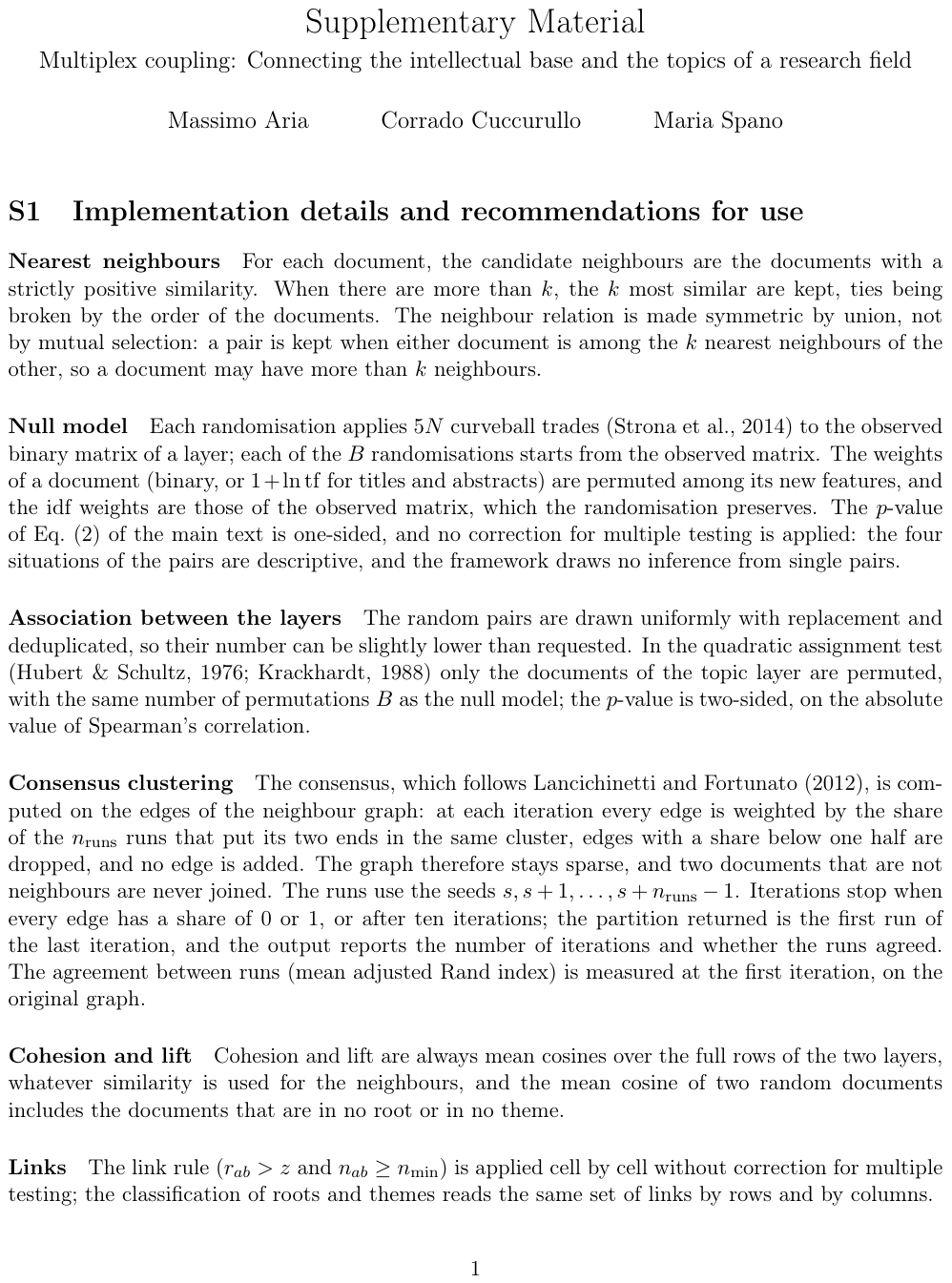}
\end{document}